\documentclass[journal]{IEEEtran}

\usepackage{cite}
\usepackage{amsmath}
\usepackage{mathptmx} 
\usepackage{amssymb}  
\usepackage{amsthm}
\usepackage{algorithmic}
\usepackage{array}
\usepackage{url}
\usepackage{graphics} 
\usepackage{epsfig} 
\usepackage{enumerate}
\usepackage{booktabs}
\usepackage[bottom]{footmisc}

\newtheorem{theorem}{Theorem}
\newtheorem{corollary}[theorem]{Corollary}
\newtheorem{proposition}[theorem]{Proposition}
\newtheorem{lemma}[theorem]{Lemma}
\newtheorem{remark}[theorem]{Remark}
\newtheorem{definition}[theorem]{Definition}

\DeclareMathOperator*{\argmax}{arg\,max}
\DeclareMathOperator*{\argmin}{arg\,min}

\title{Scheduling Policies for Minimizing Age of Information in Broadcast Wireless Networks}

\author{Authors}

\author{Igor Kadota, Abhishek Sinha, Elif Uysal-Biyikoglu, Rahul Singh and Eytan Modiano 
\thanks{The authors are with the Massachusetts Institute of Technology and with the Middle East Technical University. (e-mail: kadota@mit.edu; sinhaa@mit.edu; uelif@metu.edu.tr; rsingh12@mit.edu; modiano@mit.edu)}
\thanks{This paper was presented in part at the Allerton Conference in 2016 \cite{AoI_broadcast}.}
}

\begin{document}

\maketitle

\begin{abstract}
We consider a wireless broadcast network with a base station sending time-sensitive information to a number of clients through unreliable channels. The Age of Information (AoI), namely the amount of time that elapsed since the most recently delivered packet was generated, captures the freshness of the information. We formulate a discrete-time decision problem to find a transmission scheduling policy that minimizes the expected weighted sum AoI of the clients in the network. 

We first show that in symmetric networks a Greedy policy, which transmits the packet with highest current age, is optimal. For general networks, we develop three low-complexity scheduling policies: a randomized policy, a Max-Weight policy and a Whittle's Index policy, and derive performance guarantees as a function of the network configuration. To the best of our knowledge, this is the first work to derive performance guarantees for scheduling policies that attempt to minimize AoI in wireless networks with unreliable channels. Numerical results show that both Max-Weight and Whittle's Index policies outperform the other scheduling policies in every configuration simulated, and achieve near optimal performance. 
\end{abstract}


\section{INTRODUCTION}\label{sec.Intro}
\IEEEPARstart{A}{ge} of Information (AoI) has been receiving increasing attention in the literature \cite{AoI_update,AoI_multiple,AoI_MG1,AoI_management,AoI_path,AoI_errors,AoI_gamma,AoI_nonlinear,AoI_energy15,AoI_energy17,AoI_lazy,UpdateOrWait17,AoI_scheduling,PAoI_scheduling,AoI_cache,AoI_multiaccess,AoI_sync,AoI_design,AoI_LGFS16,AoI_LGFS17,AoI_VANET,AoI_buffer,AoI_emulation,AoI_LUPMAC}, particularly for applications that generate time-sensitive information such as position, command and control, or sensor data. An interesting feature of this performance metric is that it captures the freshness of the information \emph{from the perspective of the destination}, in contrast to the long-established packet delay, that represents the freshness of the information with respect to individual packets. In particular, AoI measures the \emph{time that elapsed since the generation of the packet that was most recently delivered to the destination}, while packet delay measures the time interval between the generation of a packet and its delivery. 

The two parameters that influence AoI are packet delay and packet inter-delivery time. In general, controlling only one is insufficient for achieving good AoI performance. For example, consider an M/M/1 queue with a low arrival rate and a high service rate. In this setting, the queue is often empty, resulting in low packet delay. Nonetheless, the AoI can still be high, since infrequent packet arrivals result in outdated information at the destination. Table \ref{tab.example} provides a numerical example of an M/M/1 queue with fixed service rate $\mu=1$ and a variable arrival rate $\lambda$. The first and third rows represent a system with a high average AoI caused by high inter-delivery time and high packet delay, respectively. The second row shows the queue at the point of minimum average AoI \cite{AoI_update}.

\begin{table}[b]
\caption{Expected delay, expected inter-delivery time and average AoI of a M/M/1 queue with $\mu=1$ and variable $\lambda$.\vspace{-0.3cm}}\label{tab.example}
\begin{center}
\begin{tabular}{cccc} \toprule
$\lambda$ & $\mathbb{E}$[delay] & $\mathbb{E}$[inter-delivery] & Average AoI \\
(pkt/sec) & (sec) & (sec) & (sec) \\
\midrule
$0.01$ & $1.01$ &  $\mathbf{100.00}$ & $\mathbf{101.00}$ \\
$\mathbf{0.53}$ & $2.13$ & $1.89$ & $3.48$ \\
$0.99$ &  $\mathbf{100.00}$ & $1.01$ & $\mathbf{100.02}$ \\
\bottomrule
\end{tabular}
\end{center}
\end{table}

\emph{A good AoI performance is achieved when packets with low delay are delivered regularly}. It is important to emphasize the difference between delivering packets regularly and providing a minimum throughput. Figure~\ref{fig.Regularity} illustrates the case of two sequences of packet deliveries that have the same throughput but different delivery regularity. In general, a minimum throughput requirement can be fulfilled even if long periods with no delivery occur, as long as those are balanced by periods of consecutive deliveries. 

\begin{figure}[b]
\begin{center}
\includegraphics[height=2.5cm]{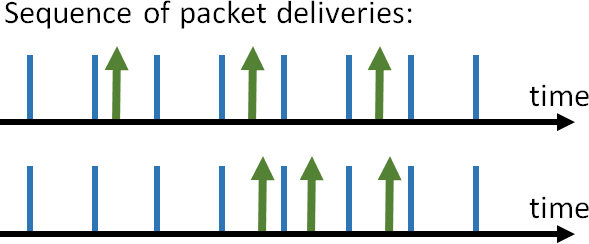}
\end{center}
\caption{Two sample sequences of packet deliveries are represented by the green arrows. Both sequences have the same throughput, namely 3 packets over the interval, but different delivery regularity.}\label{fig.Regularity}
\end{figure}

The problem of minimizing AoI was introduced in \cite{AoI_update} and has been explored using different approaches. Queueing Theory is used in \cite{AoI_update,AoI_multiple,AoI_MG1,AoI_management,AoI_path,AoI_errors,AoI_gamma,AoI_nonlinear} for finding the optimal server utilization with respect to AoI. The authors in \cite{AoI_energy15,AoI_energy17,AoI_lazy,UpdateOrWait17} consider the problem of optimizing the times in which packets are generated at the source in networks with energy-harvesting or maximum update frequency constraints. Link scheduling optimization with respect to AoI has been recently considered in \cite{AoI_scheduling,PAoI_scheduling,AoI_cache,AoI_multiaccess,AoI_sync,AoI_design,AoI_LGFS16,AoI_LGFS17}. Applications of AoI are studied in \cite{AoI_VANET,AoI_buffer,AoI_emulation,AoI_LUPMAC}.

The problem of optimizing link scheduling decisions in broadcast wireless networks with respect to throughput and delivery times has been studied extensively in the literature. Throughput maximization of traffic with strict packet delay constraints has been addressed in \cite{single_link,theoryofQoS, delay,index_schedule}. Inter-delivery time is considered in \cite{TSLS_15,TSLS_16,index_regularity,regularity_reliable,regularity_smooth,regularity_frequency,regularity_round_robin} as a measure of service regularity. Age of Information has been considered in \cite{AoI_scheduling,PAoI_scheduling,AoI_cache,AoI_multiaccess,AoI_sync,AoI_design,AoI_LGFS16,AoI_LGFS17}. 

In this paper, 
we consider a network in which packets are generated periodically and transmitted through unreliable channels. Minimizing the AoI is particularly challenging in wireless networks with unreliable channels due to transmission errors that result in packet losses. Our main contribution is the development and analysis of four low-complexity scheduling policies: a Greedy policy, a randomized policy, a Max-Weight policy and a Whittle's Index policy. We first show that Greedy achieves minimum AoI in symmetric networks. Then, for general networks, we compare the performance of each policy against the optimal AoI and derive the corresponding performance guarantees. To the best of our knowledge, this is the first work to derive performance guarantees for policies that attempt to minimize AoI in wireless networks with unreliable channels. A preliminary version of this work appeared in \cite{AoI_broadcast}.

The remainder of this paper is outlined as follows. In Sec.~\ref{sec.Model}, the network model is presented. In Sec.~\ref{sec.Symmetric}, we find the optimal scheduling policy for the case of symmetric networks. In Sec.~\ref{sec.General}, we consider the general network case and derive performance guarantees for the Greedy, Randomized and Max-Weight policies. In Sec.~\ref{sec.Whittle}, we establish that the AoI minimization problem is indexable and obtain the Whittle's Index in closed-form. Numerical results are presented in Sec.~\ref{sec.Simulation}. The paper is concluded in Sec.~\ref{sec.Conclusion}.
\section{SYSTEM MODEL}\label{sec.Model}
Consider a single-hop wireless network with a base station (BS) sending time-sensitive information to $M$ clients. Let the time be slotted, with $T$ consecutive slots forming a frame. At the beginning of every frame, the BS generates one packet per client $i \in \{1,2,\cdots,M\}$. Those new packets replace any undelivered packets from the previous frame. Denote the frame index by the positive integer $k$. Packets are periodically generated at every frame $k$ for each client $i$, thus, each packet can be unequivocally identified by the tuple $(k,i)$.

Let $n \in \{1,\cdots,T\}$ be the index of the slot within a frame. A slot is identified by the tuple $(k,n)$. In a slot, the BS transmits a packet to a selected client $i$ over the wireless channel. The packet is successfully delivered to client $i$ with probability $p_i \in (0,1]$ and a transmission error occurs with probability $1-p_i$. The probability of successful transmission $p_i$ is fixed in time, but may differ across clients. The client sends a feedback signal to the BS after every transmission. The feedback (success / failure) reaches the BS instantaneously and without errors.

The transmission scheduling policies considered in this paper are non-anticipative, i.e. policies that do not use future knowledge in selecting clients. Let $\Pi$ be the class of non-anticipative policies and $\pi \in \Pi$ be an arbitrary admissible policy. In a slot $(k,n)$, policy $\pi$ can either idle or select a client with an undelivered packet. Clients that have already received their packet by slot $(k,n)$ can only be selected in the next frame $k+1$. Scheduling policies attempt to minimize the expected weighted sum AoI of the clients in the network. Next, we discuss this performance metric.

\subsection{Age of Information Formulation}\label{sec.AoI}
Prior to introducing the expected weighted sum AoI, we characterize the Age of Information of a single client in the context of our system model. Let $AoI_i$ be the positive real number that represents the Age of Information of client $i$. The $AoI_i$ increases linearly in time when there is no delivery of packets to client $i$. At the end of the frame in which a delivery occurs, the $AoI_i$ is updated to $T$. In Fig.~\ref{fig.AoI}, the evolution of $AoI_i$ is illustrated for a given sample sequence of deliveries to client $i$.

\begin{figure}[b!]
\begin{center}
\includegraphics[height=4.9cm]{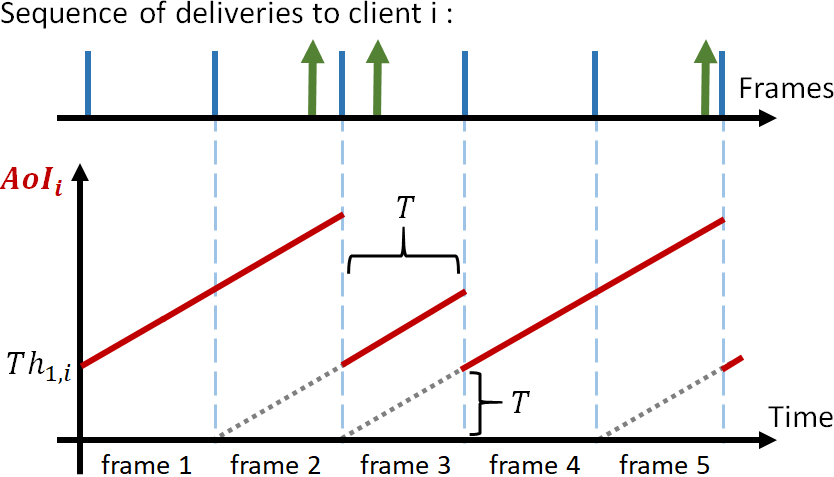}
\end{center}
\caption{On the top, a sample sequence of deliveries to client $i$ during five frames. The upward arrows represent the times of packet deliveries. On the bottom, the associated evolution of the $AoI_i$.}\label{fig.AoI}
\end{figure}

In Fig.~\ref{fig.AoI_close}, the $AoI_i$ is shown in detail. Let $\hat{s}_k$ denote the set of clients that successfully received packets during frame $k$ and let the positive integer $h_{k,i}$ represent the number of frames since the last delivery to client $i$. At the beginning of frame $k+1$, the value of $h_{k,i}$ is updated as follows
\begin{equation}\label{eq.evolution_h}
h_{k+1,i}=\left\{ \begin{array}{cl}
h_{k,i}+1 &, \mbox{ if $i \notin \hat{s}_{k}$ ; } \\
1 &, \mbox{ if $i \in \hat{s}_{k}$ . } \end{array} \right.
\end{equation}

\begin{figure}[b!]
\begin{center}
\includegraphics[height=3.2cm]{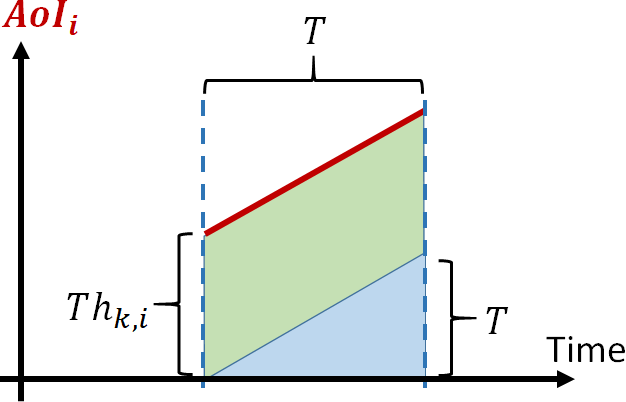}
\end{center}
\caption{Area under $AoI_i$ during any frame $k$ in terms of $h_{k,i}$ and $T$.}\label{fig.AoI_close}
\end{figure}

As can be seen in Fig.~\ref{fig.AoI_close}, during frame $k$ the area under the $AoI_i$ curve can be divided into a triangle of area $T^2/2$ and a parallelogram of area $h_{k,i}T^2$. This area, averaged over time, captures the average Age of Information associated with client~$i$. A network-wide metric for measuring the freshness of the information is the Expected Weighted Sum AoI, namely
\begin{align}
\mbox{EWSAoI}&=\frac{1}{KTM}\mathbb{E}\left[\sum_{k=1}^{K}\sum_{i=1}^{M}\alpha_i\left(\frac{T^2}{2}+T^2h_{k,i}\right) \left| \; \vec{h}_1 \right. \right]\nonumber\\
&=\frac{T}{2M}\sum_{i=1}^{M}\alpha_i +\frac{T}{KM}\mathbb{E}\left[\sum_{k=1}^{K}\sum_{i=1}^{M}\alpha_ih_{k,i}\left| \; \vec{h}_1 \right. \right] \; , \label{eq.EWSAoI}
\end{align}
where $\alpha_i$ is the positive real value that denotes the client's weight and the vector $\vec{h}_1=[h_{1,1},\cdots,h_{1,M}]^T$ represents the initial values of $h_{k,i}$ in \eqref{eq.evolution_h}. For notation simplicity, we omit $\vec{h}_1$ hereafter. Manipulating the expression of EWSAoI gives us the objective function
\begin{equation}\label{eq.Objective}
\min_{\pi \in \Pi}\mathbb{E}\left[J_K^{\pi}\right], \mbox{ where } J_K^{\pi}=\frac{1}{KM}\sum_{k=1}^{K}\sum_{i=1}^{M}\alpha_i \; h_{k,i}^\pi \; ,
\end{equation}
where \eqref{eq.Objective} is obtained by subtracting the constant terms from \eqref{eq.EWSAoI} and dividing the result by $T$. As can be seen by the relationship between \eqref{eq.EWSAoI} and \eqref{eq.Objective}, the scheduling policy that minimizes $\mathbb{E}\left[J_K^{\pi}\right]$ is the same policy that minimizes EWSAoI. Henceforth in this paper, we refer to this policy as \emph{AoI-optimal}. With the definitions of AoI\footnote{For ease of exposition, in this paper, the value of $AoI_i$ is updated at the beginning of the frame that follows a successful transmission to client $i$, rather than immediately after the successful transmission. This update mechanism simplifies the problem while maintaining the features of interest.} and objective function presented, in the next section we introduce the Greedy policy. Table~\ref{tab.notation} summarizes key notation.

\begin{table}[b]
\caption{Description of key notation.\vspace{-0.3cm}}\label{tab.notation}
\begin{center}
\begin{tabular}{cl} \toprule
$M$ & number of clients. Client index is $i \in \{1,2,\cdots,M\}$ \\
$K$ & number of frames. Frame index is $k \in \{1,2,\cdots,K\}$ \\
$T$ & number of slots in a frame. Slot index is $n \in \{1,2,\cdots,T\}$ \\
$p_i$ & probability of successful transmission to client $i$ \\
$\pi$ & admissible non-anticipative scheduling policy \\
$AoI_i$ & Age of Information associated with client $i$ \\
$h_{k,i}$ & number of frames since the last packet delivery to client $i$ \\
$\hat{s}_k$ & set of clients that received a packet during frame $k$ \\
$\alpha_i$ & weight of client $i$. Represents the relative importance of $AoI_i$ \\
$\mathbb{E}[J_K^\pi]$ & objective function that represents the performance of policy $\pi$ \\
$L_B$ & Lower Bound on $\mathbb{E}[J_K^\pi]$ for any admissible policy $\pi$\\
$U_B^\pi$ & Upper Bound on $\mathbb{E}[J_K^\pi]$ for a particular policy $\pi$ \\
$\rho^\pi$ & AoI guarantee associated with policy $\pi$ \\
$D_i(K)$ & number of packet deliveries to client $i$ up to frame $K$ \\
$A_i(K)$ & number of packet transmissions to client $i$ up to frame $K$ \\
$I_i[m]$ & number of frames between consecutive deliveries to client $i$ \\
$R_i$ & number of frames remaining after the last delivery to client $i$ \\
$\bar{\mathbb{M}}[.]$ & operator that calculates the sample mean of a set of values \\
$\bar{\mathbb{V}}[.]$ & operator that calculates the sample variance of a set of values \\
\bottomrule
\end{tabular}
\end{center}
\end{table}



\section{OPTIMALITY OF GREEDY}\label{sec.Symmetric}
In this section, we introduce the Greedy policy and show that it minimizes the AoI of the finite-horizon scheduling problem described in Sec.~\ref{sec.Model} under some conditions on the underlying network. The Greedy policy is defined next.

\emph{Greedy policy schedules in each slot $(k,n)$ a transmission to the client with highest value of $h_{k,i}$ that has an undelivered packet, with ties being broken arbitrarily.}

Denote the Greedy policy as $G$. Observe that Greedy is non-anticipative and work-conserving, i.e. it only idles after all packets have been delivered during frame $k$. 
Next, we discuss a few properties of the Greedy policy that lead to the optimality result in Theorem~\ref{theo.Greedy}.

\begin{remark}\label{rem.Greedy}
The Greedy policy switches scheduling decisions only after a successful packet delivery.
\end{remark}
In slot $(k,n)$, Greedy selects client $i=\argmax_{j} \{ h_{k,j} \}$ from the set of clients with an undelivered packet. Assume that this packet transmission fails and the subsequent slot is in the same frame $k$. Since $\vec{h}_{k}$ remains unchanged and client $i$ still has an undelivered packet, the Greedy policy selects the same client $i$. Alternatively, if the next slot is in frame $k+1$, then $\vec{h}_{k+1,i}$ evolves according to \eqref{eq.evolution_h} and Greedy selects $\argmax_{j} \{ h_{k+1,j} \}$ from the set of all clients. It follows from \eqref{eq.evolution_h} that client $i$ is selected again. Hence, the Greedy policy selects the same client $i$, uninterruptedly, until its packet is delivered.

\begin{lemma}[Round Robin]\label{lem.Greedy}
Without loss of generality, reorder the client index $i$ in descending order of $\vec{h}_1$, with client $1$ having the highest $h_{1,i}$ and client $M$ the lowest $h_{1,i}$. The Greedy policy \textbf{delivers} packets according to the index sequence $(1,2,\cdots,M,1,2,\cdots)$ until the end of the time-horizon $K$, i.e. Greedy follows a Round Robin pattern.
\end{lemma}

The proof of Lemma~\ref{lem.Greedy} is in Appendix \ref{app.Lemma_Greedy}. Together, Remark~\ref{rem.Greedy} and Lemma~\ref{lem.Greedy}, provide a complete description of the behavior of Greedy. Consider a network with $\vec{h}_1$ reordered as in Lemma~\ref{lem.Greedy}, the Greedy policy schedules client $1$, repeatedly, until one packet is delivered, then it schedules client $2$, repeatedly, until one packet is delivered, and so on, following the Round Robin pattern until the end of the time-horizon. The Greedy policy only idles when all $M$ packets are delivered in the same frame. Figure \ref{fig.Greedy} illustrates a sequence of scheduling decisions of Greedy in a network with error-free channels. 

\begin{figure}[b]
\begin{center}
\includegraphics[height=2.1cm]{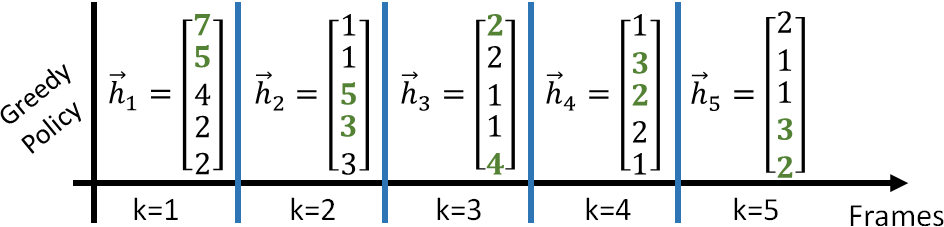}
\end{center}
\caption{Evolution of $\vec{h}_k$ when the Greedy policy is employed in a network with $M=5$ clients, $T=2$ slots per frame, error-free channels, $p_i=1, \forall i$, and $\vec{h}_{1}=[7~5~4~2~2]^T$. In each frame, the Greedy policy transmits packets of two clients. The elements of $\vec{h}_{k}$ associated with the clients that received a packet during frame $k$ are depicted in bold green. All elements in $\vec{h}_{k}$ change according to \eqref{eq.evolution_h}: green elements are updated to 1 while black elements are incremented by 1. In this figure, the Round Robin pattern is evident.}\label{fig.Greedy}
\end{figure}

\begin{corollary}[Steady-State of Greedy for error-free channels]\label{cor.Greedy}
Consider a network with error-free channels, $p_i=1, \forall i$. The Greedy policy drives this network to a steady-state in which the sum of the elements of $\vec{h}_k$ is constant. Let $m_1 \in \mathbb{N}$ and $m_2 \in \{0,1,\cdots,T-1\}$ be the quotient and remainder of the division of $M$ by $T$, namely $M=m_1T+m_2$. The steady-state is achieved at the beginning of frame $k=m_1+2$ and the sum of $\vec{h}_k$ is given by
\begin{equation}\label{eq.Greedy_steady}
\sum_{i=1}^M h_{k,i} = \frac{Tm_1\left(m_1+1\right)}{2}+m_2(m_1+1) \; .
\end{equation}
\end{corollary}

Corollary~\ref{cor.Greedy} follows directly from the proof of Lemma \ref{lem.Greedy} in Appendix~\ref{app.Lemma_Greedy}. The sum in \eqref{eq.Greedy_steady} comes from the expression of $\vec{h}_{k}$ in \eqref{eq.vectorh_k}. Notice that \eqref{eq.Greedy_steady} is independent of the initial $\vec{h}_1$. Figure~\ref{fig.Greedy} represents a network with $M=5$, $T=2$, $m_1=2$ and $m_2=1$. Thus, according to Corollary~\ref{cor.Greedy}, the steady-state is achieved in frame $k=4$ and the sum of the elements of $\vec{h}_k$ is $9$ for $k \geq 4$. Those values can be easily verified in Fig.~\ref{fig.Greedy}. 

In Theorem~\ref{theo.Greedy}, we establish that Greedy is AoI-optimal when the underlying network is \emph{symmetric}, namely all clients have the same channel reliability $p_i=p \in (0,1]$ and weight $\alpha_i=\alpha \geq 0$. Prior to the main result, we establish in Lemma~\ref{prop.Greedy} that Greedy is AoI-optimal for a symmetric network with error-free channels.

\begin{lemma}[Optimality of Greedy for error-free channels]\label{prop.Greedy}
Consider a symmetric network with error-free channels $p_i=1$ and weights $\alpha_i=\alpha>0, \forall i$. Among the class of admissible policies $\Pi$, the Greedy policy attains the minimum sum AoI \eqref{eq.EWSAoI}, namely
\begin{equation}\label{eq.Objective_Symmetric}
J_K^G \leq J_K^\pi, \forall \pi \in \Pi \; . 
\end{equation}
\end{lemma}

The proof of Lemma \ref{prop.Greedy} is in Appendix \ref{app.Prop_Greedy}. Intuitively, Greedy minimizes $\sum_{i=1}^M h_{k,i}$ by reducing the highest elements of $\vec{h}_{k}$ to \emph{unity} at every frame. Together, Lemma~\ref{prop.Greedy} and Corollary~\ref{cor.Greedy} show that, when channels are error-free, Greedy drives the network to a steady-state \eqref{eq.Greedy_steady} that is AoI-optimal. Next, we use the result in Lemma \ref{prop.Greedy} to show that the Greedy policy is optimal for any symmetric network.

\begin{theorem}[Optimality of Greedy]\label{theo.Greedy}
Consider a symmetric network with channel reliabilities $p_i=p \in (0,1]$ and weights $\alpha_i=\alpha > 0, \forall i$. Among the class of admissible policies $\Pi$, the Greedy policy attains the minimum expected sum AoI \eqref{eq.EWSAoI}, namely $G=\argmin_{\pi \in \Pi}\mathbb{E}\left[J_K^{\pi}\right]$.
\end{theorem}

To show that the Greedy policy minimizes the AoI of any symmetric network, we generalize Lemma \ref{prop.Greedy} using a \emph{stochastic dominance} argument \cite{coupling} that compares the evolution of $\vec{h}_k$ when Greedy is employed to that when an arbitrary policy $\pi$ is employed. The proof of Theorem \ref{theo.Greedy} is in Appendix~\ref{app.Theo_Greedy} of the supplementary material.

Selecting the client with an undelivered packet and highest value of $h_{k,i}$ in every slot is AoI-optimal for every symmetric network. For general networks, with clients possibly having different channel reliabilities $p_i$ and weights $\alpha_i$, scheduling decisions based exclusively on $\vec{h}_k$ may not be AoI-optimal.
In the next section, we develop three low-complexity scheduling policies and derive performance guarantees for every policy in the context of general networks. 
\section{AGE OF INFORMATION GUARANTEES}\label{sec.General}
One possible approach for finding a policy that minimizes the EWSAoI is to optimize the objective function in \eqref{eq.Objective} using Dynamic Programming \cite{dynProg}. A negative aspect of this approach is that evaluating the optimal scheduling decision for each state of the network can be computationally demanding, especially for networks with a large number of clients\footnote{Vector $\vec{h}_{k}=[h_{k,1},\cdots,h_{k,M}]^T$ is part of the state space of the network. Since each element $h_{k,i}$ can take at least $k$ different values, $h_{k,i}\in\{1,2,\cdots,k\}$, the set of possible values of $\vec{h}_{k}$ has cardinality at least $k^M$, implying that the state space grows exponentially with the number of clients $M$.}. To overcome this problem, known as the curse of dimensionality, and gain insight into the minimization of the Age of Information, we consider four low-complexity scheduling policies, namely Greedy, Randomized, Max-Weight and Whittle's Index policies, and derive performance guarantees for each of them.

For a given network setup $(M,K,T,p_i,\alpha_i)$, the performance of an arbitrary admissible policy $\pi \in \Pi$ is given by $\mathbb{E}\left[J_K^{\pi}\right]$ from \eqref{eq.Objective} and the optimal performance is $\mathbb{E}\left[J^*\right]=\min_{\eta \in \Pi}\mathbb{E}\left[J_K^{\eta}\right]$. Ideally, when expressions for $\mathbb{E}\left[J_K^{\pi}\right]$ and $\mathbb{E}\left[J^*\right]$ are available, we define the optimality ratio $\psi^\pi:= \mathbb{E}\left[J_K^{\pi}\right] / \mathbb{E}\left[J^*\right]$ and say that policy $\pi$ is $\psi^\pi$-optimal. Naturally, the closer $\psi^\pi$ is to one, the better is the performance of policy $\pi$. Alternatively, when expressions for $\mathbb{E}\left[J_K^{\pi}\right]$ and $\mathbb{E}\left[J^*\right]$ are not available, we define the ratio 
\begin{equation}
\rho^\pi:=\frac{U_B^\pi}{L_B} \; ,
\end{equation}
where $L_B$ is a lower bound to the AoI-optimal performance and $U_B^\pi$ is an upper bound to the performance of policy $\pi$. It follows from the inequality $L_B \leq \mathbb{E}\left[J^*\right] \leq \mathbb{E}\left[J_K^{\pi}\right] \leq  U_B^\pi $ that $\psi^\pi\leq\rho^\pi$ and thus we can say that policy $\pi$ is $\rho^\pi$-optimal. 

Next, we obtain a lower bound $L_B$ that is used for deriving performance guarantees $\rho^\pi$ for the four low-complexity scheduling policies of interest. Henceforth in this section, \emph{we consider the infinite-horizon problem} where $K \rightarrow \infty$. The focus on the long-term behavior of the system allows us to derive simpler and more insightful performance guarantees.


\subsection{Universal Lower Bound}\label{sec.LowerBound}
In this section, we find a lower bound to the solution of the objective function in \eqref{eq.Objective}.
\begin{theorem}[Lower Bound]\label{theo.LowerBound}
For a given network setup, we have $L_B \leq \lim_{K \rightarrow \infty}\mathbb{E}\left[J_K^\pi\right], \; \forall \pi \in \Pi$, where
\begin{equation}\label{eq.LowerBound}
L_B = \frac{1}{2MT} \left( \sum_{i=1}^{M} \sqrt{\frac{\alpha_i}{p_i}} \right)^2 + \frac{1}{2M}\sum_{i=1}^{M}\alpha_i \; . 
\end{equation}
\end{theorem}

\begin{proof}
First, we use a sample path argument to characterize the evolution of $\vec{h}_k$ over time. Then, we derive an expression for the objective function of the infinite-horizon problem, namely $\lim_{K \rightarrow \infty}J_K^\pi$, and manipulate this expression to obtain $L_B$ in \eqref{eq.LowerBound}. Fatou's lemma is employed to establish the result in Theorem~\ref{theo.LowerBound}.

Consider a sample path $\omega \in \Omega$ associated with a scheduling policy $\pi \in \Pi$ and a finite time-horizon $K$. For this sample path, let $D_i(K)$ be the total number of packets delivered to client $i$ up to and including frame $K$, let $I_i[m]$ be the number of frames between the $(m-1)$th and $m$th deliveries to client $i$, i.e. the inter-delivery times of client $i$, and let $R_i$ be the number of frames remaining after the last packet delivery to the same client. Then, the time-horizon can be written as follows
\begin{equation}\label{eq.horizon_division}
K=\sum_{m=1}^{D_i(K)} I_i[m] + R_i, \forall i \in \{1,2,\cdots,M\} \; .
\end{equation}

The evolution of $h_{k,i}$ is well-defined in each of the time intervals $I_i[m]$ and $R_i$. During the frames associated with the interval $I_i[m]$, the parameter $h_{k,i}$ evolves as $1,2,\cdots,I_i[m]$. During the frames associated with the interval $R_i$, the value of $h_{k,i}$ evolves as $1,2,\cdots,R_i$. Hence, the objective function in \eqref{eq.Objective} can be rewritten as
\begin{align}
J_K^{\pi}=&\frac{1}{KM}\sum_{k=1}^{K}\sum_{i=1}^{M}\alpha_i \; h_{k,i}=\frac{1}{M}\sum_{i=1}^{M}\frac{\alpha_i}{K}\left[\sum_{k=1}^{K}h_{k,i}\right]\\
=&\frac{1}{M}\sum_{i=1}^{M}\frac{\alpha_i}{K}\left[\sum_{m=1}^{D_i(K)} \frac{(I_i[m]+1)I_i[m]}{2} + \frac{(R_i+1)R_i}{2} \right]\nonumber\\
\overset{(a)}{=}&\frac{1}{2M}\sum_{i=1}^{M}\frac{\alpha_i}{K}\left[\sum_{m=1}^{D_i(K)} I_i^2[m]+R_i^2+K \right]\nonumber\\
=&\frac{1}{2M}\sum_{i=1}^{M}\alpha_i\left[\frac{D_i(K)}{K} \left(\frac{1}{D_i(K)}\sum_{m=1}^{D_i(K)}I_i^2[m]\right)+\frac{R_i^2}{K}+1 \right] \; , \nonumber
\end{align}
where (a) uses \eqref{eq.horizon_division} to substitute the sum of the linear terms $I_i[m]$ and $R_i$ by $K$.

Now, define the operator $\bar{\mathbb{M}}[.]$ that calculates the sample mean of a set of values. Using this operator, let the sample mean of $I_i[m]$ and $I_i^2[m]$ be
\begin{align}
\bar{\mathbb{M}}[I_i]&=\frac{1}{D_i(K)}\sum_{m=1}^{D_i(K)} I_i[m] \label{eq.sample_mean} \; ;\\
\bar{\mathbb{M}}[I_i^2]&=\frac{1}{D_i(K)}\sum_{m=1}^{D_i(K)} I_i^2[m] \; . \label{eq.sample_mean_2}
\end{align}
Combining \eqref{eq.horizon_division} and \eqref{eq.sample_mean} yields
\begin{equation}\label{eq.relation_throughput}
\frac{K}{D_i(K)}=\frac{\sum_{j=1}^{D_i(K)} I_i[j] + R_i}{D_i(K)}=\bar{\mathbb{M}}[I_i] + \frac{R_i}{D_i(K)} \; .
\end{equation}
Substituting \eqref{eq.sample_mean_2} and \eqref{eq.relation_throughput} into the objective function gives
\begin{equation}\label{eq.AoI_expression}
J_K^\pi=\frac{1}{2M}\sum_{i=1}^{M}\alpha_i\left[\left[\bar{\mathbb{M}}[I_i] + \frac{R_i}{D_i(K)}\right]^{-1}\bar{\mathbb{M}}[I_i^2]+\frac{R_i^2}{K}+1 \right] \, ,
\end{equation}
with probability one.

To simplify \eqref{eq.AoI_expression}, consider the infinite-horizon problem with $K \rightarrow \infty$ and assume that the admissible class $\Pi$ \emph{does not contain policies that starve clients}. 
\begin{definition}
A policy $\pi$ starves client $i$ if, with a positive probability, it stops transmitting packets to that client after frame $K'<\infty$. 
\end{definition}
When $\pi$ starves client $i$, the expected number of frames after the last packet delivery is $\mathbb{E}\left[R_i\right] \rightarrow \infty$ and the objective function $\mathbb{E}\left[J_K^{\pi}\right] \rightarrow \infty$. Therefore, policies that starve clients are excluded from the class $\Pi$ without loss of optimality.


Since policies in $\Pi$ transmit packets to every client repeatedly and each packet transmission has a positive probability $p_i$ of being delivered, it follows that $I_i[m]$ and $R_i$ are finite with probability one. Thus, in the limit $K \rightarrow \infty$, we have $R_i^2/K \rightarrow 0$, $D_i(K) \rightarrow \infty$ and $R_i/D_i(K) \rightarrow 0$. Applying those limits to $J_K^\pi$ in \eqref{eq.AoI_expression} gives the \emph{objective function of the infinite-horizon AoI problem}
\begin{equation}\label{eq.Objective_infiniteHorizon}
\lim_{K \rightarrow \infty}J_K^{\pi} =\frac{1}{2M}\sum_{i=1}^{M}\alpha_i\left[\frac{\bar{\mathbb{M}}[I_i^2]}{\bar{\mathbb{M}}[I_i]}+1 \right] \quad \mbox{w.p.1} \; .
\end{equation}
This insightful expression depicts the relationship between AoI and the moments of the inter-delivery time $I_i[m]$. 

Prior to deriving the expression of $L_B$ in \eqref{eq.LowerBound}, we introduce some useful quantities. Define the operator $\bar{\mathbb{V}}[.]$ that calculates the sample variance of a set of values. Let the sample variance of $I_i[m]$ be 
\begin{equation}\label{eq.sample_variance}
\bar{\mathbb{V}}[I_i]=\frac{1}{D_i(K)}\sum_{m=1}^{D_i(K)} \left( I_i[m]-\bar{\mathbb{M}}[I_i] \right)^2 \; .
\end{equation}
Notice that the sample variance is positive valued and $\bar{\mathbb{V}}[I_i]= \bar{\mathbb{M}}[I_i^2]-\left(\bar{\mathbb{M}}[I_i]\right)^2$. Let $A_i(K)$ be the total number of packets \emph{transmitted} to client $i$ up to and including frame $K$. Any policy $\pi$ can schedule at most one client per slot, hence
\begin{equation}\label{eq.relation_AI1}
\sum_{i=1}^{M}A_i(K) \leq KT \quad \mbox{w.p.1} \; .
\end{equation}
Moreover, since every transmission to client $i$ is delivered with the same probability $p_i$, independently of the outcome of previous transmissions, by the strong law of large numbers
\begin{equation}\label{eq.relation_AI2}
\lim_{K \rightarrow \infty}  \;\frac{D_i(K)}{A_i(K)}=p_i \quad \mbox{w.p.1} \; .
\end{equation}

With the definitions of $\bar{\mathbb{V}}[I_i]$ and $A_i(K)$, we obtain $L_B$ by manipulating the objective function of the infinite-horizon AoI problem in \eqref{eq.Objective_infiniteHorizon} as follows
\begin{align}\label{eq.LowerBound_1}
\lim_{K \rightarrow \infty}J_K^{\pi} = & \frac{1}{2M}\sum_{i=1}^{M}\alpha_i\left[\frac{\bar{\mathbb{V}}[I_i]}{\bar{\mathbb{M}}[I_i]}+\bar{\mathbb{M}}[I_i]+1 \right] \nonumber\\
\overset{(a)}{\geq} & \frac{1}{2M}\sum_{i=1}^{M}\alpha_i \bar{\mathbb{M}}[I_i] + \frac{1}{2M}\sum_{i=1}^{M}\alpha_i \nonumber\\
\overset{(b)}{=}&\lim_{K \rightarrow \infty}\frac{1}{2MT} KT \sum_{i=1}^M \frac{\alpha_i}{D_i(K)} + \frac{1}{2M}\sum_{i=1}^{M}\alpha_i \nonumber\\
\overset{(c)}{\geq} &\lim_{K \rightarrow \infty}\frac{1}{2MT} \left(\sum_{j=1}^{M}A_j(K)\right) \left(\sum_{i=1}^M \frac{\alpha_i}{D_i(K)}\right)+ \frac{1}{2M}\sum_{i=1}^{M}\alpha_i \nonumber\\
\overset{(d)}{\geq} &\lim_{K \rightarrow \infty}\frac{1}{2MT} \left( \sum_{i=1}^{M} \sqrt{\frac{\alpha_i A_i(K)}{D_i(K)}} \right)^2 + \frac{1}{2M}\sum_{i=1}^{M}\alpha_i \nonumber\\ 
\overset{(e)}{=}&\frac{1}{2MT} \left( \sum_{i=1}^{M} \sqrt{\frac{\alpha_i}{p_i}} \right)^2 + \frac{1}{2M}\sum_{i=1}^{M}\alpha_i \quad w.p.1 \; ,
\end{align}
where (a) uses the fact that $\bar{\mathbb{V}}[I_i]\geq 0$, (b) uses \eqref{eq.relation_throughput} with $K \rightarrow \infty$, (c) uses the inequality in \eqref{eq.relation_AI1}, (d) uses Cauchy-Schwarz inequality and (e) uses the equality in \eqref{eq.relation_AI2}.  Notice that \eqref{eq.LowerBound_1} gives the expression for $L_B$ found in \eqref{eq.LowerBound}.

Finally, since $J_K^\pi$ in \eqref{eq.AoI_expression} is positive for every $\pi\in\Pi$ and for every $K$, we employ Fatou's lemma to \eqref{eq.LowerBound_1} and obtain $\lim_{K \rightarrow \infty}\mathbb{E}\left[J_K^{\pi}\right] \geq \mathbb{E}\left[\lim_{K \rightarrow \infty}J_K^{\pi}\right] \geq L_B$, establishing the result of the theorem.
\end{proof}

The sequence of inequalities in \eqref{eq.LowerBound_1} that led to $\lim_{K \rightarrow \infty}\mathbb{E}\left[J_K^{\pi}\right] \geq L_B$ could have rendered a loose lower bound. However, in the next section, we use $L_B$ to derive a performance guarantee $\rho^G$ for the Greedy policy and show that $\rho^G=1$ for symmetric networks with large $M$, i.e. under these conditions the value of $L_B$ is as tight as possible. Furthermore, numerical results in Sec.~\ref{sec.Simulation} show that the lower bound is also tight in other network configurations. In the upcoming sections, we obtain performance guarantees for four low-complexity scheduling policies: Greedy, Randomized, Max-Weight and Whittle's Index policies.

\subsection{Greedy Policy}\label{sec.Greedy}
In this section, we analyze the Greedy policy introduced in Sec.~\ref{sec.Symmetric} and derive a closed-form expression for its performance guarantee $\rho^G$. The expression for $\rho^G$ depends on the statistics of the set of values $\{1/p_i\}_{i=1}^M$, in particular of its \emph{coefficient of variation}. 
Let the sample mean and sample variance of $\{1/p_i\}_{i=1}^M$ be
\begin{align}
\bar{\mathbb{M}}\left[\frac{1}{p_i}\right]&=\frac{1}{M}\sum_{j=1}^M\frac{1}{p_j} \; ;\\
\bar{\mathbb{V}}\left[\frac{1}{p_i}\right]&=\frac{1}{M}\sum_{j=1}^M \left(\frac{1}{p_j}-\bar{\mathbb{M}}\left[\frac{1}{p_i}\right]\right)^2 \; .
\end{align}
Then, the coefficient of variation is given by
\begin{equation}\label{eq.CV_p}
C_V=\displaystyle\sqrt{\bar{\mathbb{V}}\left[\frac{1}{p_i}\right]}\bigg/\bar{\mathbb{M}}\left[\frac{1}{p_i}\right] \; .
\end{equation}
The coefficient of variation is a measure of how spread out are the values of $1/p_i$. The value of $C_V$ is large when $\{1/p_i\}_{i=1}^M$ are disperse and $C_V=0$ if and only if $p_i=p$ for all clients.


\begin{theorem}[Performance of Greedy]\label{theo.performance_Greedy}
Consider a network $(M,T,p_i,\alpha_i)$ with an infinite time-horizon. The Greedy policy is $\rho^G$-optimal as $M\rightarrow\infty$, where
\begin{equation}\label{eq.performance_Greedy}
\rho^G=\frac{\displaystyle\left(\sum_{i=1}^M\alpha_i\right)\left(\sum_{i=1}^M \frac{1}{p_i}\right)\left[1+\frac{C_V^2}{M}\right]+T\left(\sum_{i=1}^M\alpha_i\right)}{\displaystyle\left(\sum_{i=1}^M \sqrt{\frac{\alpha_i}{p_i}}\right)^2+T\left(\sum_{i=1}^M\alpha_i\right)} \; .
\end{equation}
\end{theorem}

The proof of Theorem~\ref{theo.performance_Greedy} is in Appendix~\ref{app.Theo_performance_Greedy} of the supplementary material. The expression of $\rho^G$ for finite $M$ can be readily obtained by dividing \eqref{eq.performance_Greedy_M} by \eqref{eq.LowerBound}. 
Next, we use the performance guarantee in \eqref{eq.performance_Greedy} to obtain sufficient conditions for the optimality of the Greedy policy.

\begin{corollary}\label{cor.performance_Greedy}
The Greedy policy minimizes the expected sum AoI \eqref{eq.EWSAoI} of any symmetric network with $M\rightarrow\infty$. 
\end{corollary}

\begin{proof}
Consider two inequalities. (i) Cauchy-Schwarz
\begin{equation}
\left(\sum_{i=1}^M \sqrt{\frac{\alpha_i}{p_i}}\right)^2 \leq \left(\sum_{i=1}^M\alpha_i\right)\left(\sum_{i=1}^M \frac{1}{p_i}\right) \; ,
\end{equation}
and (ii) Positive coefficient of variation: $C_V \geq 0$. It is evident from \eqref{eq.performance_Greedy} that $\rho^G=1$ if and only if both inequalities (i) and (ii) hold with \emph{equality} and this is true if and only if $\alpha_i=\alpha$ and $p_i=p$ for all clients. 
\end{proof}

Theorem~\ref{theo.performance_Greedy} provides a closed-form expression for the performance guarantee $\rho^G$ and Corollary \ref{cor.performance_Greedy} shows that, by leveraging the knowledge of $h_{k,i}$, the Greedy policy achieves optimal performance in symmetric networks with $M\rightarrow\infty$. Notice that the Greedy policy does not take into account differences in terms of weight $\alpha_i$ and channel reliability $p_i$. 
In the next section, we study the class of Stationary Randomized policies which use the knowledge of $\alpha_i$ and $p_i$ but neglect $h_{k,i}$.

\subsection{Stationary Randomized Policy}\label{sec.Randomized}
Consider the class of Stationary Randomized policies in which scheduling decisions are made randomly, according to fixed probabilities. In particular, define the Randomized policy as follows.

\emph{Randomized policy selects in each slot $(k,n)$ client $i$ with probability $\beta_i/\sum_{j=1}^M\beta_j$, for every client $i$ and for positive fixed values of $\{\beta_i\}_{i=1}^M$. The BS transmits the packet if the selected client has an undelivered packet and idles otherwise.}

Denote the Randomized policy as $R$. Observe that this policy uses no information from current or past states of the network. Moreover, it is not work-conserving, since the BS can idle when the network still has clients with undelivered packets. Next, we derive a closed-form expression for the performance guarantee $\rho^R$ and find a Randomized policy that is $2$-optimal \emph{for all network configurations with} $T=1$ slot per frame.

\begin{theorem}[Performance of Randomized]\label{theo.performance_Random}
Consider a network $(M,T,p_i,\alpha_i)$ with an infinite time-horizon. The Randomized policy with positive values of $\{\beta_i\}_{i=1}^M$ is $\rho^R$-optimal, where
\begin{equation}\label{eq.performance_Random}
\rho^R=2\frac{\displaystyle\left(\sum_{j=1}^{M}\beta_j\sum_{i=1}^{M}\frac{\alpha_i }{p_i\beta_i}\right)+(T-1)\left(\sum_{i=1}^{M}\frac{\alpha_i}{p_i}\right)}{\displaystyle\left(\sum_{i=1}^M \sqrt{\frac{\alpha_i}{p_i}}\right)^2+T\left(\sum_{i=1}^M\alpha_i\right)} \; .
\end{equation}
\end{theorem}

\begin{proof}
The performance guarantee is defined as $\rho^R=U_B^R/L_B$, where the denominator is the universal lower bound in \eqref{eq.LowerBound} and the numerator is an upper bound to the objective function, namely $\lim_{K \rightarrow \infty} \mathbb{E}[J_K^R] \leq U_B^R$, which is derived in Appendix \ref{app.Theo_performance_Random} of the supplementary material. 

Let $d_i(k)\in \{0,1\}$ be the number of packets delivered to client $i$ during frame $k$. Notice that 
\begin{equation}
\mathbb{E}\left[d_i(k)\right]=\mathbb{P}(\mbox{delivery to client $i$ during frame $k$}) \; .
\end{equation} 
When the Randomized policy is employed, this probability is constant over time, i.e. $\mathbb{E}\left[d_i(k)\right]=\mathbb{E}\left[d_i\right]$. Moreover, the PMF of the random variable $I_i[m]$ that represents the number of frames between the $(m-1)$th and $m$th packet deliveries to client $i$ is given by
\begin{equation}
\mathbb{P}\left(I_i[m]=n\right)=\mathbb{E}\left[d_i\right](1-\mathbb{E}\left[d_i\right])^{n-1} \; ,
\end{equation}
for $n \in \{1,2,\cdots\}$ and is \emph{independent of} $m$. 

Clearly, when the Randomized policy is employed, the sequence of packet deliveries is a renewal process with geometric inter-delivery times $I_i[m]$. Thus, using the generalization of the elementary renewal theorem for renewal-reward processes \cite[Sec.~5.7]{DSP} yields
\begin{equation}\label{eq.Random_Renewal}
\lim_{K\rightarrow\infty}\frac{1}{K}\sum_{k=1}^K\mathbb{E}[h_{k,i}]=\frac{\mathbb{E}[I_i[m]^2]}{2\mathbb{E}[I_i[m]]}+\frac{1}{2} = \frac{1}{\mathbb{E}[d_i]} \; ,
\end{equation}
and substituting \eqref{eq.Random_Renewal} into the objective function \eqref{eq.Objective} gives
\begin{align}\label{eq.Objective_Random}
\lim_{K \rightarrow \infty}\mathbb{E}\left[J_K^{R}\right] &=\frac{1}{M}\sum_{i=1}^{M}\frac{\alpha_i}{\mathbb{E}\left[d_i\right]} \; .
\end{align}

For simplicity of exposition, we consider the case $T=1$ slot per frame. The derivation of the performance guarantee $\rho^R$ for general $T$ is in Appendix~\ref{app.Theo_performance_Random}. When $T=1$, packets are always available for transmission and the BS selects one client per frame. Hence, the probability of delivering a packet to client $i$ during frame $k$ is 
\begin{equation}\label{eq.prob_Random_T1}
\mathbb{E}\left[d_i\right]=\frac{\beta_i}{\sum_{j=1}^M\beta_j}p_i \; .
\end{equation}
Substituting \eqref{eq.prob_Random_T1} into \eqref{eq.Objective_Random} gives 
\begin{align}\label{eq.UB_Random_T1}
\lim_{K \rightarrow \infty}\mathbb{E}\left[J_K^{R}\right]=\frac{1}{M}\sum_{j=1}^{M}\beta_j\sum_{i=1}^{M}\frac{\alpha_i }{p_i\beta_i} =U_B^R\; .
\end{align}
Finally, dividing \eqref{eq.UB_Random_T1} by the lower bound in \eqref{eq.LowerBound} gives the performance guarantee $\rho^R$ in \eqref{eq.performance_Random} for $T=1$.
\end{proof}

\begin{corollary}\label{cor.performance_Random}
The Randomized policy with $\beta_i=\sqrt{\alpha_i/p_i}, \forall i$, has $\rho^R<2$ for all networks with $T=1$ slot per frame.
\end{corollary}
\begin{proof}
The assignment $\beta_i=\sqrt{\alpha_i/p_i}, \forall i \in \{1,\cdots,M\}$ is the necessary (and sufficient) condition for the Cauchy-Schwarz inequality 
\begin{equation}
\left(\sum_{i=1}^M \sqrt{\frac{\alpha_i}{p_i}}\right)^2 \leq \left(\sum_{j=1}^M\beta_j\right)\left(\sum_{i=1}^M \frac{\alpha_i}{\beta_ip_i}\right) \; ,
\end{equation}
to hold with equality. Applying this condition to the expression in \eqref{eq.performance_Random} for $T=1$ results in $\rho^R<2$.
\end{proof}

Theorem~\ref{theo.performance_Random} gives an expression for $\rho^R$ and Corollary~\ref{cor.performance_Random} shows that, by using only the knowledge of $\alpha_i$ and $p_i$, a Randomized policy can achieve 2-optimal performance in a \emph{wide range of network setups}, in particular all networks with $T=1$ slot per frame. Next, we develop a Max-Weight policy that leverages the knowledge of $\alpha_i$, $p_i$ and $h_{k,i}$ in making scheduling decisions.

\subsection{Max-Weight Policy}\label{sec.MaxWeight}
In this section, we use concepts from Lyapunov Optimization \cite{lyapunov} to derive a Max-Weight policy. The Max-Weight policy is obtained by minimizing the drift of a Lyapunov Function of the system state at every frame $k$. Consider the quadratic Lyapunov Function
\begin{equation}\label{eq.Lyapunov_function}
L(\vec{h}_k)=\frac{1}{M}\sum_{i=1}^M\alpha_ih_{k,i}^2 \; , 
\end{equation}
and the one-frame Lyapunov Drift
\begin{equation}\label{eq.def_Lyapunov_drift}
\Delta(\vec{h}_k)=\mathbb{E}\left[\left. L(\vec{h}_{k+1})-L(\vec{h}_k)\right|\vec{h}_{k}\right] \, .
\end{equation}
The Lyapunov Function $L(\vec{h}_k)$ depicts how large the AoI of the clients in the network during frame $k$ is, while the Lyapunov Drift $\Delta(\vec{h}_k)$ represents the growth of $L(\vec{h}_k)$ from one frame to the next. Intuitively, by minimizing the drift, the Max-Weight policy reduces the value of $L(\vec{h}_k)$ and, consequently, keeps the AoI of the clients low. 

To find the policy that minimizes the one-frame drift $\Delta(\vec{h}_k)$, we first need to analyze the RHS of \eqref{eq.def_Lyapunov_drift}. Consider frame $k$ with a fixed vector $\vec{h}_k$ and a policy $\pi$ making scheduling decisions throughout the $T$ slots of this frame. Recall that $d_i^\pi(k)\in\{0,1\}$ represents the number of packets delivered to client $i$ during frame $k$ when policy $\pi$ is employed. An alternative way to represent the evolution of $h_{k,i}$ defined in \eqref{eq.evolution_h} is
\begin{equation}\label{eq.h_MW}
h_{k+1,i}=d_i^\pi(k)+(h_{k,i}+1)[1-d_i^\pi(k)] \; .
\end{equation}
Applying \eqref{eq.h_MW} into the conditional expectation of $h_{k+1,i}^2$ yields
\begin{align}\label{eq.h_MW2}
\mathbb{E}\left[h_{k+1,i}^2-h_{k,i}^2|\vec{h}_{k}\right]=-\mathbb{E}\left[d_i^\pi(k)|\vec{h}_{k}\right]h_{k,i}(h_{k,i}+2)+2h_{k,i}+1 \; .
\end{align}
Substituting \eqref{eq.Lyapunov_function} into \eqref{eq.def_Lyapunov_drift} and then using \eqref{eq.h_MW2} gives the following expression for the Lyapunov Drift 
\begin{align}
\Delta(\vec{h}_k)=&-\frac{1}{M}\sum_{i=1}^M\mathbb{E}\left[d_i^\pi(k)|\vec{h}_{k}\right]\alpha_ih_{k,i}(h_{k,i}+2)+\nonumber \\
&+\frac{2}{M}\sum_{i=1}^M\alpha_ih_{k,i}+\frac{1}{M}\sum_{i=1}^M\alpha_i \; . \label{eq.Lyapunov_drift}
\end{align}

Observe that the scheduling policy $\pi$ only affects the first term on the RHS of \eqref{eq.Lyapunov_drift}. Define the weight function $G_i(h_{k,i})=\alpha_ih_{k,i}(h_{k,i}+2)$. During frame $k$, the scheduling policy that maximizes the sum $\sum_{i=1}^M\mathbb{E}\left[d_i^\pi(k)|\vec{h}_{k}\right]G_i(h_{k,i})$ also minimizes $\Delta(\vec{h}_k)$. Notice that $\mathbb{E}\left[d_i^\pi(k)|\vec{h}_{k}\right]$ represents the expected throughput of client $i$ during frame $k$. The class of policies that maximize the \emph{expected weighted sum throughput in a frame} was studied in \cite{theoryofQoS,multicast}. According to \cite[Eq.(2)]{multicast}, to maximize $\sum_{i=1}^M\mathbb{E}\left[d_i^\pi(k)|\vec{h}_{k}\right]G_i(h_{k,i})$, the scheduling policy must myopically select the client with an undelivered packet and highest value of $p_iG_i(h_{k,i})$ in every slot of frame $k$. Hence, the Max-Weight policy is defined as follows.

\emph{Max-Weight policy schedules in each slot $(k,n)$ a transmission to the client with highest value of $p_i\alpha_ih_{k,i}(h_{k,i}+2)$ that has an undelivered packet, with ties being broken arbitrarily.}

Denote the Max-Weight policy as $MW$. Observe that when $\alpha_i=\alpha$ and $p_i=p$, prioritizing according to $p_i\alpha_ih_{k,i}(h_{k,i}+2)$ is identical to prioritizing according to $h_{k,i}$, i.e. Max-Weight is identical to Greedy. Thus, from Theorem~\ref{theo.Greedy} (Optimality of Greedy), we conclude that Max-Weight is AoI-optimal for symmetric networks. For general networks, we derive the performance guarantee $\rho^{MW}$ for the Max-Weight policy.

\begin{theorem}[Performance of Max-Weight]\label{theo.performance_MaxWeight}
Consider a network $(M,T,p_i,\alpha_i)$ with an infinite time-horizon. The Max-Weight policy is $\rho^{MW}$-optimal, where
\begin{equation}\label{eq.performance_MaxWeight}
\rho^{MW}=4 \frac{\displaystyle\left(\sum_{i=1}^{M}\sqrt{\frac{\alpha_i}{p_i}}\right)^2+(T-1)\sum_{i=1}^M\frac{\alpha_i}{p_i}}{\displaystyle\left(\sum_{i=1}^{M}\sqrt{\frac{\alpha_i}{p_i}}\right)^2+T\left(\sum_{i=1}^M\alpha_i\right)} \; .
\end{equation}
\end{theorem}

The proof of Theorem~\ref{theo.performance_MaxWeight} is in Appendix~\ref{app.Theo_performance_MaxWeight} of the supplementary material. In contrast to the Greedy and Randomized policies, the Max-Weight policy uses all available information, namely $p_i$, $\alpha_i$ and $h_{k,i}$, in making scheduling decisions. As expected, numerical results in Sec.~\ref{sec.Simulation} demonstrate that Max-Weight outperforms both Greedy and Randomized in every network setup simulated. In fact, the performance of Max-Weight is comparable to the optimal performance computed using Dynamic Programming. However, by comparing the performance guarantee $\rho^{MW}$ in \eqref{eq.performance_MaxWeight} with $\rho^G$ and $\rho^R$, it might seem that Max-Weight does not provide better performance. The reason for this is the challenge to obtain a tight performance upper bound for Max-Weight. As opposed to Greedy and Randomized, the Max-Weight policy cannot be evaluated using Renewal Theory and it does not have properties that simplify the analysis, such as packets being delivered following a Round Robin pattern or clients being selected according to fixed probabilities. 
Next, we consider the AoI minimization problem from a different perspective and propose an Index policy \cite{RMAB}, also known as Whittle's Index policy. This policy is surprisingly similar to the Max-Weight policy and also yields a strong performance. 

\section{Whittle's Index Policy}\label{sec.Whittle}
Whittle's Index policy is the optimal solution to a relaxation of the Restless Multi-Armed Bandit (RMAB) problem. This low-complexity heuristic policy has been extensively used in the literature \cite{index_regularity,index_schedule,index_myopic} and is known to have a strong performance in a range of applications \cite{index_multichannelaccess,index_assympt}. The challenge associated with this approach is that the Index policy is only defined for problems that are \emph{indexable}, a condition which is often difficult to establish. 

To develop the Whittle's Index policy, the AoI minimization problem is transformed into a relaxed RMAB problem. The first step is to note that each client in the AoI problem evolves as a restless bandit. Thus, the AoI problem can be posed as a RMAB problem. The second step is to consider the relaxed version of the RMAB problem, called the Decoupled Model, in which clients are examined separately. The Decoupled Model associated with each client $i$ adheres to the network model with $M=1$, except for the addition of a \emph{service charge}. The service charge is a fixed cost per transmission $C>0$ that is incurred by the network every time the BS transmits a packet. 
The last step is to solve the Decoupled Model. This solution lays the foundation for the design of the Index policy. Next, we formulate and solve the Decoupled Model, establish that the AoI problem is indexable and derive the Whittle's Index policy. A detailed introduction to the Whittle's Index policy can be found in \cite{RMAB,RMAB_book}. \\

\subsection{Decoupled Model}\label{sec.Decoupled_Model}
The Decoupled Model is formulated as a Dynamic Program (DP). For presenting the cost-to-go function, which is central to the DP, we first introduce the state, control, transition and objective of the model. Then, using the expression of the cost-to-go, we establish in Proposition~\ref{prop.Decoupled} a key property of the Decoupled Model which is used in the characterization of its optimal scheduling policy. Since the Decoupled Model considers only a single client, hereafter in this section, we omit the client index $i$.

Consider the network model from Sec. \ref{sec.Model} with $M=1$ client. Recall that at the beginning of every frame, the BS generates a new packet that replaces any undelivered packet from previous frame. Let $s_{k,n}$ represent the delivery status of this new packet at the beginning of slot $(k,n)$. If the packet has been successfully delivered to the client by the beginning of slot $(k,n)$, then $s_{k,n}=1$, and if the packet is still undelivered, $s_{k,n}=0$. The tuple $(s_{k,n},h_{k})$ depicts the system state, for it provides a complete characterization of the network at slot $(k,n)$. 

Denote by $u_{k,n}$ the scheduling decision in time slot $(k,n)$. This quantity is equal to $1$ if the BS transmits the packet in slot $(k,n)$, and $u_{k,n}=0$ otherwise. Since the BS can only transmit undelivered packets, if $s_{k,n}=1$, then the decision must be to idle $u_{k,n}=0$.

State transitions are different at frame boundaries and within frames. At the boundary between frames $k-1$ and $k$, namely, in the transition from slot $(k-1,T)$ to slot $(k,1)$, each component of the system state $(s_{k,n},h_k)$ evolves in a distinct way. Since the BS generates a new packet at the beginning of slot $(k,1)$, we have $s_{k,1}=0$ for every frame $k$. Whereas, the evolution of $h_k$ is divided into two cases: i) case $u_{k-1,T}=1$, when the BS transmits the packet during slot $(k-1,T)$, the value of $h_k$ depends on the feedback signal, as follows
\begin{align}
P(h_k=h_{k-1}+1|h_{k-1})=1-p \; ; & \quad\mbox{[failure]} \\
P(h_k=1|h_{k-1})=p \; ; & \quad\mbox{[success]} 
\end{align}
and ii) case $u_{k-1,T}=0$, when the BS idles, the transition is deterministic
\begin{align}
P(h_k=h_{k-1}+1|h_{k-1})=1 \; , & \quad\mbox{if $s_{k-1,T}=0$ ;} \\
P(h_k=1|h_{k-1})=1 \; , & \quad\mbox{if $s_{k-1,T}=1$ .}
\end{align}

For state transitions that occur within the same frame, the quantity $h_k$ remains fixed and $s_{k,n}$ evolves
according to the scheduling decisions and feedback signals. If the BS idles during slot $(k,n-1)$, the delivery status of the packet does not change, thus
\begin{equation} 
P(s_{k,n}=s_{k,n-1}|s_{k,n-1})=1 \; .
\end{equation}
If the BS transmits during slot $(k,n-1)$, the value of $s_{k,n}$ depends upon the outcome of the transmission, as given by 
\begin{align}
P(s_{k,n}=0|s_{k,n-1})=1-p \; ; & \quad\mbox{[failure]} \\
P(s_{k,n}=1 |s_{k,n-1})=p \; . & \quad\mbox{[success]} 
\end{align}

The last concept to be discussed prior to the cost-to-go function is the objective. The objective function of the Decoupled Model, $\mathcal{J}_K^{\pi}$, is analogous to $J_K^\pi$ in \eqref{eq.Objective}, except that it represents a single client, introduces the service charge $C$ and evolves in slot increments (instead of frame increments). The expression for the objective function is given by
\begin{align}\label{eq.Objective_decoupled}
&\min_{\pi \in \Pi}\mathbb{E}\left[\mathcal{J}_K^{\pi}\right] \; , \\ 
\mbox{ where } &\mathcal{J}_K^{\pi}=\frac{1}{KT}\sum_{k=1}^{K}\sum_{n=1}^{T}\left( \alpha \; h_k + C \; u_{k,n} \right) \; . \nonumber 
\end{align}




The cost-to-go function $\mathcal{J}_{k,n}(s_{k,n},h_k)$ associated with the optimization problem in \eqref{eq.Objective_decoupled} has two forms. For the \emph{last} slot of any frame $k$, namely slot $(k,T)$, the cost-to-go is expressed as
\begin{align}\label{eq.recurssion1}
\mathcal{J}_{k,T}(s_{k,T}&,h_k) = \alpha h_k + \nonumber\\
+&\min_{u_{k,T}\in\{0,1\}} \left\{ C \; u_{k,T} +  \mathbb{E}[\mathcal{J}_{k+1,1}(0,h_{k+1})] \right\} \; ,
\end{align}
and for slots \emph{other than the last}, we have
\begin{align}\label{eq.recurssion2}
\mathcal{J}_{k,n}(s_{k,n}&,h_k)=\alpha h_k + \nonumber\\
+&\min_{u_{k,n}\in\{0,1\}} \left\{ C \; u_{k,n} +  \mathbb{E}[\mathcal{J}_{k,n+1}(s_{k,n+1},h_k)] \right\} \; .
\end{align}


Given a network setup $(K,T,p,\alpha,h_1,C)$, it is possible to use backward induction on \eqref{eq.recurssion1} and \eqref{eq.recurssion2} to compute the optimal scheduling policy $\pi^*$ for the Decoupled Model. However, for the purpose of designing the Index policy, it is not sufficient to provide an algorithm that computes the optimal policy. The Index policy is based on a complete characterization of $\pi^*$. Proposition~\ref{prop.Decoupled} provides a key feature of the optimal scheduling policy which is used in its characterization.



\begin{proposition}\label{prop.Decoupled}
Consider the Decoupled Model and its optimal scheduling policy $\pi^*$. During any frame $k$, the optimal policy either: (i) idles in every slot; or (ii) transmits until the packet is delivered or the frame ends.
\end{proposition}

\begin{proof}
The proof follows from the analysis of the backward induction algorithm on \eqref{eq.recurssion1} and \eqref{eq.recurssion2}. For this proof, we assume that the algorithm has been running and that the values of $\mathcal{J}_{k+1,1}(s_{k+1,1},h_{k+1})$ for all possible system states are known. The proof is centered around the backward induction during frame $k$ and for a fixed value of $h_k$.


First, we analyze the (trivial) case in which the packet has already been delivered by the beginning of slot $(k,n)$, i.e. $s_{k,n}=1$. In this case, the optimal scheduling policy always idles. 

For the more interesting case of an undelivered packet, we start by analyzing the last slot of the frame, namely slot $(k,T)$. It follows from the cost-to-go in \eqref{eq.recurssion1} that the optimal scheduling decision $u_{k,T}^*$ depends only on the expression
\begin{equation}\label{eq.expression_Decoupled}
C-p \left[ \mathcal{J}_{k+1,1}(0,h_k+1) - \mathcal{J}_{k+1,1}(0,1)\right] \; .
\end{equation}
The optimal policy idles in slot $(k,T)$ if \eqref{eq.expression_Decoupled} is non-negative and transmits if \eqref{eq.expression_Decoupled} is negative. 
By analyzing the cost-to-go function in \eqref{eq.recurssion2}, which is associated with the optimal scheduling decisions in the remaining slots of frame $k$, it is possible to use mathematical induction to establish that:
\begin{itemize}
\item if it is optimal to transmit in slot $(k,n+1)$, then it is also optimal to transmit in slot $(k,n)$; and
\item if it is optimal to idle in slot $(k,n+1)$, then it is also optimal to idle in slot $(k,n)$.
\end{itemize}

We conclude that if \eqref{eq.expression_Decoupled} is non-negative, the optimal policy idles in every slot of frame $k$, and if \eqref{eq.expression_Decoupled} is negative, the optimal policy transmits until the packet is delivered or until frame $k$ ends. 
\end{proof}

Let $\Gamma \subset \Pi$ be the subclass of all scheduling policies that satisfy Proposition~\ref{prop.Decoupled}. Since the optimal policy is such that $\pi^* \in \Gamma$, we can reduce the scope of the Decoupled Model to policies in $\Gamma$ without loss of optimality. In the following section, we redefine the Decoupled Model so that scheduling decisions are made only once per frame, rather than once per slot. This new model is denoted Frame-Based Decoupled Model.\\


\subsection{Frame-Based Decoupled Model}\label{sec.Frame_Decoupled}
Denote by $u_{k}$ the scheduling decision at the beginning of frame $k$. We let $u_{k}=0$ if the BS idles in every slot of frame $k$ and $u_{k}=1$ if the BS transmits repeatedly until the packet is delivered or the frame ends. 

Since this discrete-time decision problem evolves in frames and every frame begins with $s_{k,1}=0$, we can fully represent the system state by $h_k$. State transitions follow the evolution of $h_{k}$ in \eqref{eq.evolution_h} and can be divided into two cases: i) case $u_{k-1}=0$, when the BS idles during frame $k-1$
\begin{equation}
P(h_k=h_{k-1}+1|h_{k-1})=1 \; ,
\end{equation}
and ii) case $u_{k-1}=1$, when the BS transmits, the state transition depends on whether the packet was delivered or discarded during frame $k-1$, as follows
\begin{align}
P(h_k=h_{k-1}+1|h_{k-1})=(1-p)^T \; ; & \quad\mbox{[discarded]}  \\
P(h_k=1|h_{k-1})=1-(1-p)^T \; . & \quad\mbox{[delivered]}  
\end{align}

The objective function of the Frame-Based Decoupled Model, $\hat{\mathcal{J}}_K^{\pi}$, is given by
\begin{equation}\label{eq.Objective_decoupled_frame}
\min_{\pi \in \Gamma}\mathbb{E}\left[\hat{\mathcal{J}}_K^{\pi}\right] \; , \mbox{ where } \hat{\mathcal{J}}_K^{\pi}=\frac{1}{KT}\sum_{k=1}^{K}\left( T \alpha \; h_k + \hat{C} \; u_{k} \right) \; ,
\end{equation}
and $\hat{C}=C(1-(1-p)^T)/p$ is the expected value of the service charge incurred during a frame in which the BS transmits. By construction, the Frame-Based Decoupled Model is equivalent to the Decoupled Model when the optimization is carried over the policies in $\Gamma$. Thus, both models have the same optimal scheduling policy $\pi^*\in\Gamma\subset\Pi$. Next, we characterize $\pi^*$ for the infinite-horizon problem.


Consider the Frame-Based Decoupled Model over an infinite-horizon with $K \rightarrow \infty$. The state and control of the system in \emph{steady-state} are denoted $h$ and $u$, respectively. Then, Bellman equations are given by $S(1)=0$ and
\begin{align}\label{eq.Bellman}
S(h)+&\lambda T=\min\{ T \alpha h+S(h+1) \; ; \\ &\hat{C} + T \alpha h + (1-p)^TS(h+1) +(1-(1-p)^T)S(1) \} \; , \nonumber
\end{align}
for all $h \in \{1,2,\cdots\}$, where $\lambda$ is the optimal average cost and $S(h)$ is the differential cost-to-go function. 
Notice that the upper part of the minimization in \eqref{eq.Bellman} is associated with choosing $u=0$, i.e. idling in every slot of the frame, and the lower part with $u=1$, i.e. transmitting until the packet is delivered or the frame ends, with ties being broken in favor to idling. 
The stationary scheduling policy that solves Bellman equations\footnote{In general, Expected Average Cost problems over an infinite-horizon and with \emph{countably infinite state space} are challenging to address. For the Frame-Based Decoupled Model, it can be shown that \cite[Proposition 5.6.1]{dynProg} is satisfied under some additional conditions on $\Gamma$. The results in \cite[Proposition 5.6.1]{dynProg} and Proposition~\ref{prop.Threshold} are sufficient to establish the optimality of the stationary scheduling policy $\pi^*$.} is given in Proposition~\ref{prop.Threshold}.


\begin{proposition}[Threshold Policy]\label{prop.Threshold}
Consider the Frame-Based Decoupled Model over an infinite-horizon. The stationary scheduling policy $\pi^*$ that solves Bellman equations \eqref{eq.Bellman} is a threshold policy in which the BS transmits during frames that have $h > H-1$ and idles when $1\leq h\leq H-1$, where the threshold $H$ is given by
\begin{equation}\label{eq.H}
H=\left\lfloor 1-Z + \sqrt{Z^2+\frac{2C}{pT\alpha}} \right\rfloor \; ,
\end{equation}
and the value of $Z$ is
\begin{equation}
Z=\frac{1}{2}+\frac{(1-p)^T}{(1-(1-p)^T)} \; .
\end{equation}
\end{proposition}

The proof of Proposition~\ref{prop.Threshold} is in Appendix~\ref{app.Prop_Threshold} of the supplementary material. Intuitively, we expect that the optimal scheduling decision is to transmit during frames in which $h$ is high (attempting to reduce the value of $h$) and to idle when $h$ is low (avoiding the service charge $\hat{C}$). Moreover, if the optimal decision is to transmit when the state is $h=H$, it is natural to expect that for all $h \geq H$ the optimal decision is also to transmit. This behavior characterizes a threshold policy. In Appendix~\ref{app.Prop_Threshold}, we demonstrate this behavior and find the minimum integer $H$ for which the optimal decision is to transmit. With the complete characterization of $\pi^*$ provided in Proposition~\ref{prop.Threshold}, we have the necessary background to establish indexability and to obtain the Whittle's Index policy for the AoI minimization problem.\\

\subsection{Indexability and Index Policy}\label{sec.Decoupled_Index}
Consider the Decoupled Model and its optimal scheduling policy $\pi^*$. Let $\mathcal{P}(C)$ be the set of states $h$ for which it is optimal to idle when the service charge is $C$, i.e. $\mathcal{P}(C)=\{h \in \mathbb{N} | h < H\}$. Note from \eqref{eq.H} that the threshold $H$ is a function of $C$. The definition of indexability is given next.

\begin{definition}[Indexability] The Decoupled Model associated with client $i$ is indexable if $\mathcal{P}(C)$ increases monotonically from $\emptyset$ to the entire state space, $\mathbb{N}$, as the service charge $C$ increases from $0$ to $+\infty$. The AoI minimization problem is indexable if the Decoupled Model is indexable for all clients $i$.
\end{definition}

The indexability of the Decoupled Model follows directly from the expression of $H$ in \eqref{eq.H}. Clearly, the threshold $H$ is monotonically increasing with $C$. Also, substituting $C=0$ yields $H=1$, which implies $\mathcal{P}(C)=\emptyset$, and the limit $C \rightarrow +\infty$ gives $H \rightarrow +\infty$ and, consequently, $\mathcal{P}(C)=\mathbb{N}$. Since this is true for the Decoupled Model associated with every client $i$, we conclude that the AoI minimization problem is indexable. Prior to introducing the Index policy, we define the Whittle's Index.

\begin{definition}[Index] Consider the Decoupled Model and denote by $C(h)$ the Whittle's Index in state $h$. Given indexability, $C(h)$ is the infimum service charge $C$ that makes both scheduling decisions (idle, transmit) equally desirable in state $h$.
\end{definition}
The closed-form expression for $C(h)$ comes from the fact that, for both scheduling decisions to be equally desirable in state $h$, the threshold must be $H=h+1$. Substituting $H=h+1$ into \eqref{eq.H} and isolating $C$, gives
\begin{equation}
C(h)=p \alpha h \left[ h + \frac{1+(1-p)^T}{1-(1-p)^T} \right] \; .
\end{equation}

After establishing indexability and finding the closed-form expression for the Whittle's Index, we return to our original problem, with the BS transmitting packets to $M$ clients. Recall that there is no service charge in the original problem. The Whittle's Index policy is described next. 

\emph{Whittle's Index policy schedules in each slot $(k,n)$ a transmission to the client with highest value of}
\begin{equation}\label{eq.index}
C_i(h_{k,i})=p_i \alpha_i h_i \left[ h_i + \frac{1+(1-p_i)^T}{1-(1-p_i)^T} \right] \; ,
\end{equation}
\emph{that has an undelivered packet, with ties being broken arbitrarily.} 

Denote the Whittle's Index policy as $WI$. By construction, the index $C_i(h_{k,i})$ represents the service charge that the network would be willing to pay in order to transmit a packet to client $i$ during frame $k$. Intuitively, by selecting the client with highest $C_i(h_{k,i})$, the Whittle's Index policy is transmitting the most valuable packet. Note that the Whittle's Index policy is similar to the Max-Weight policy despite the fact that they were developed using different methods. Both the Whittle's Index and Max-Weight policies have strong performances and both are equivalent to the Greedy policy when the network is symmetric, implying that $WI$ and $MW$ are AoI-optimal when $\alpha_i=\alpha$ and $p_i=p$. Next, we derive the performance guarantee $\rho^{WI}$ for the Whittle's Index policy.

\begin{theorem}[Performance of Whittle]\label{theo.performance_Whittle}
Consider a network $(M,T,p_i,\alpha_i)$ with an infinite time-horizon. The Whittle's Index policy is $\rho^{WI}$-optimal, where
\begin{equation}\label{eq.performance_Whittle}
\rho^{WI}=4 \frac{\displaystyle\left(\sum_{i=1}^{M}\sqrt{\frac{\mathbf{\widetilde{\alpha}_i}}{p_i}}\right)^2+(T-1)\sum_{i=1}^M\frac{\mathbf{\widetilde{\alpha}_i}}{p_i}}{\displaystyle\left(\sum_{i=1}^{M}\sqrt{\frac{\alpha_i}{p_i}}\right)^2+T\left(\sum_{i=1}^M\alpha_i\right)} \; ,
\end{equation}
and 
\begin{equation}
\mathbf{\widetilde{\alpha}_i} = \frac{\alpha_i}{2}\left(\frac{2}{1-(1-p_i)^T}+1\right)^2\; .
\end{equation}
\end{theorem}

To find the expression for the performance guarantee of the Whittle's Index policy $\rho^{WI}$ in \eqref{eq.performance_Whittle}, we use similar arguments to the ones for deriving $\rho^{MW}$. The proof of Theorem~\ref{theo.performance_Whittle} is in Appendix~\ref{app.Theo_performance_Whittle} of the supplementary material. Next, we evaluate the performance of the four low-complexity scheduling policies discussed in this paper using MATLAB simulations.
\section{SIMULATION RESULTS}\label{sec.Simulation}
In this section, we evaluate the performance of the scheduling policies in terms of the Expected Weighted Sum Age of Information in \eqref{eq.EWSAoI}. We compare five scheduling policies: i) Greedy policy; ii) Randomized policy with $\beta_i=\sqrt{\alpha_i/p_i}$; iii) Max-Weight policy; iv) Whittle's Index policy and v) the optimal Dynamic Program. The numerical results associated with the first four policies are \emph{simulations}, while the results associated with the Dynamic Program are \emph{computations} of EWSAoI obtained by applying Value Iteration to the objective function \eqref{eq.Objective}. By definition, the Dynamic Program yields the optimal performance.

Figs.~\ref{fig.Sym} and \ref{fig.Gen} evaluate the scheduling policies in a variety of network settings. In Fig. \ref{fig.Sym}, we consider a two-user \emph{symmetric} network with $T=5$ slots in a frame, a total of $K=150$ frames and both clients having the same weight $\alpha_1=\alpha_2=1$ and channel reliability $p_1=p_2 \in \{1/15,\cdots,14/15\}$. In Fig. \ref{fig.Gen}, we consider a two-user \emph{non-symmetric} network with $K=200$, $p_1=2/3$, $p_2=1/10$, $T \in \{1,\cdots,10\}$ and both clients having $\alpha_1=\alpha_2=1$. The initial vector is $\vec{h}_1=[1,1,\cdots,1]^T$ in all simulations. 

\begin{figure}[b!]
\begin{center}
\includegraphics[height=5.2cm]{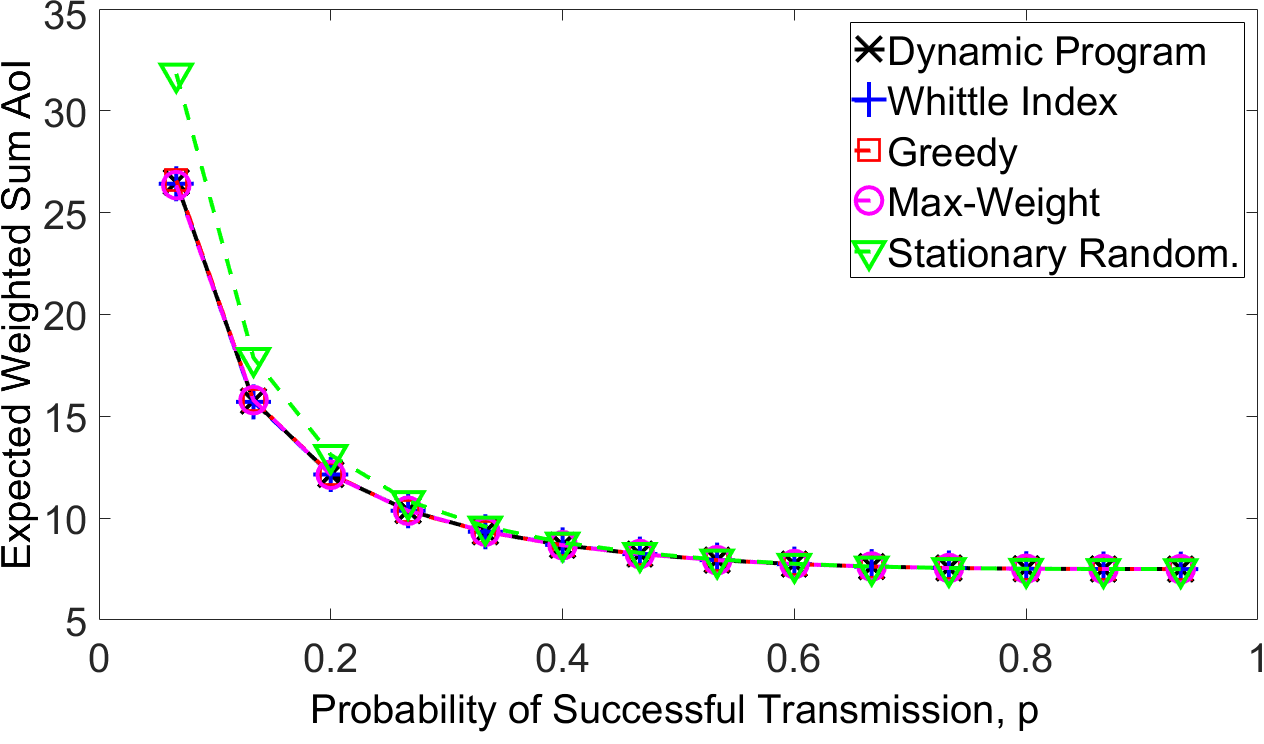}
\end{center}
\caption{Two-user symmetric network with $T=5, K=150, \alpha_i=1, p_i=p, \;\forall i$. The simulation result for each policy and for each value of $p$ is an average over $1,000$ runs.}\label{fig.Sym}
\end{figure}

\begin{figure}[b!]
\begin{center}
\includegraphics[height=5.2cm]{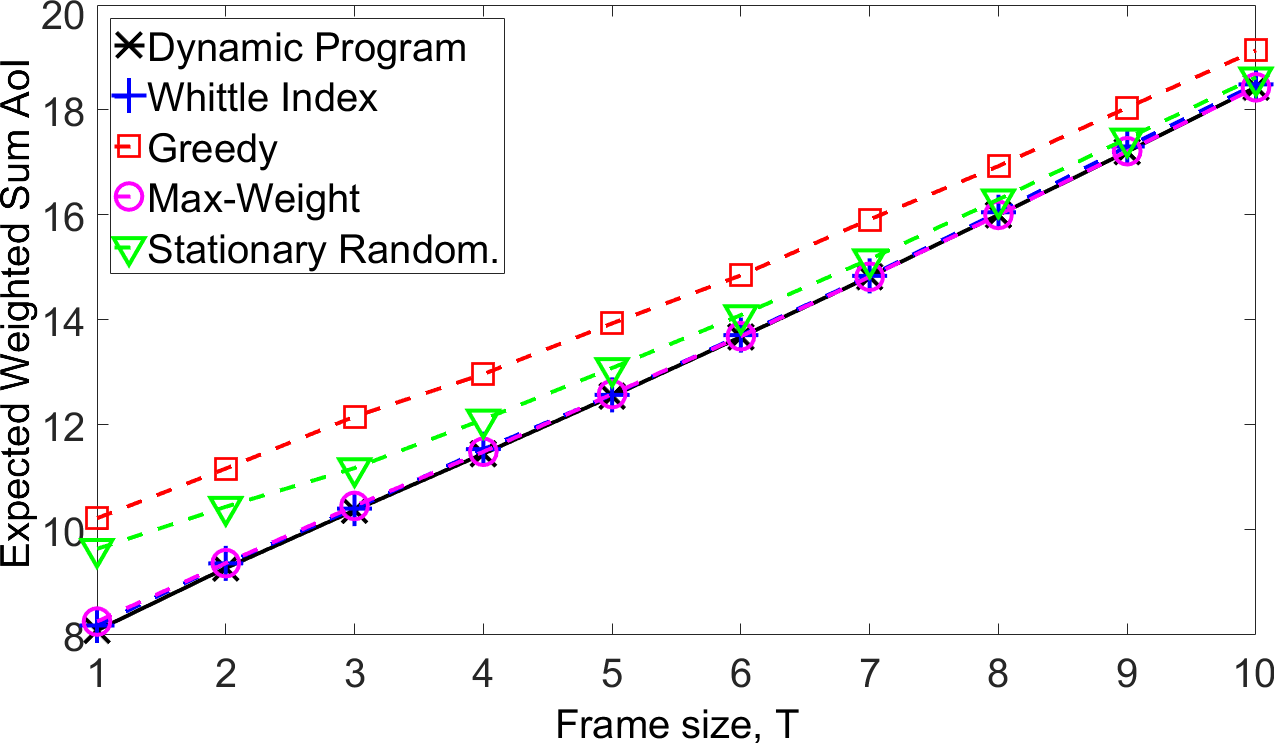}
\end{center}
\caption{Two-user general network with $K=200, \alpha_i=1, p_1=2/3, p_2=1/10, \;\forall i$. The simulation result for each policy and for each value of $T$ is an average over $1,000$ runs.}\label{fig.Gen}
\end{figure}

Figs.~\ref{fig.Gen_client} and \ref{fig.Gen_index} display the performance of the scheduling policies for larger networks. Due to the high computation complexity associated with the Dynamic Program, we show the Lower Bound $L_B$ from \eqref{eq.LowerBound} instead. In Fig.~\ref{fig.Gen_client}, we consider a network with an increasing number of clients $M \in \{5,10,\cdots,45,50\}$, $T=2$ slots in a frame, a total of $K=50,000$ frames, channel reliability $p_i=i/M, \;\forall i \in \{1,2,\cdots,M\}$ and all clients having the same weight $\alpha_i=1$. 

\begin{figure}[b!]
\begin{center}
\includegraphics[height=4.9cm]{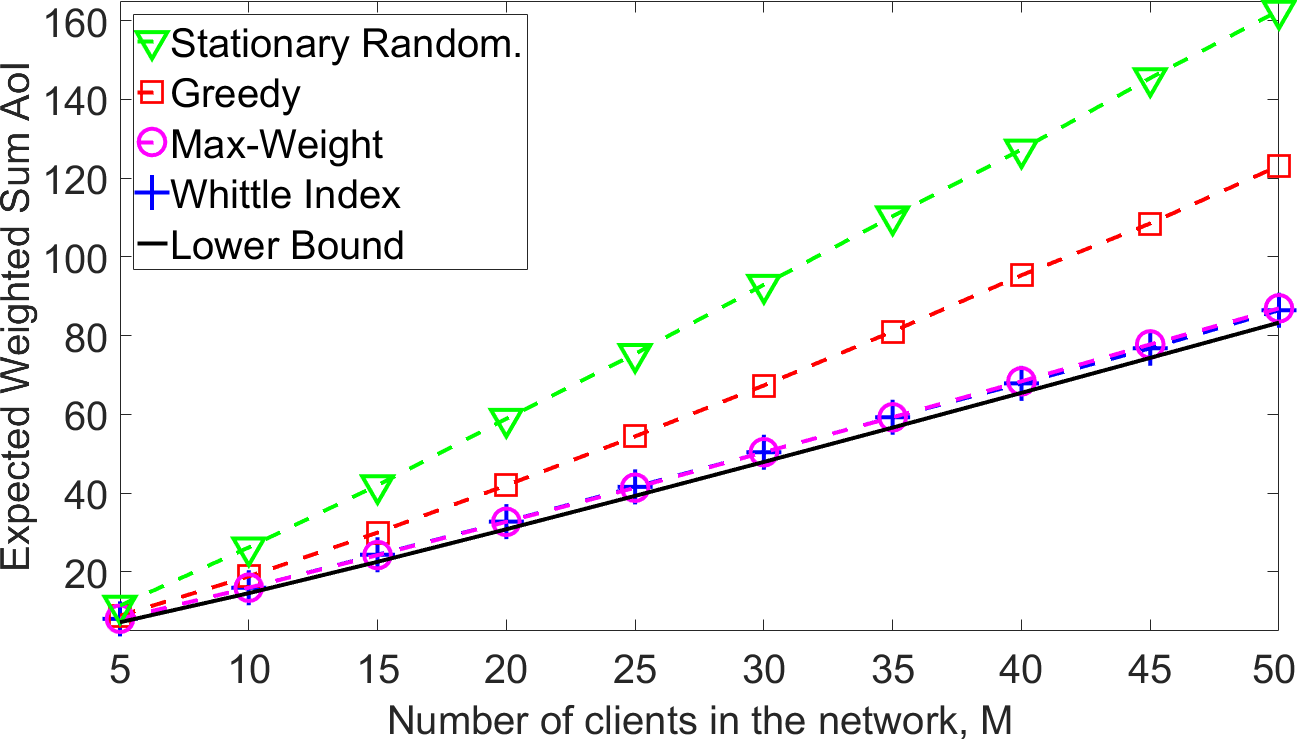}
\end{center}
\caption{Network with $T=2, K=50,000, \alpha_i=1, p_i=i/M, \;\forall i$. The simulation result for each policy and for each value of $M$ is an average over $10$ runs.}\label{fig.Gen_client}
\end{figure}

In Fig.~\ref{fig.Gen_index}, we consider a network with $M=4$ clients, $T=2$ slots in a frame, a total of $K=50,000$ frames, identical client weights $\alpha_i=1, \forall i \in \{1,2,3,4\}$ and channel reliabilities $\{p_i\}_{i=1}^M$ chosen uniformly at random in the interval $(0.1)$. A total of $2,000$ different choices of $\{p_i\}_{i=1}^M$ are considered. Network setups are displayed in ascending order of $L_B$.

\begin{figure}[b!]
\begin{center}
\includegraphics[height=5cm]{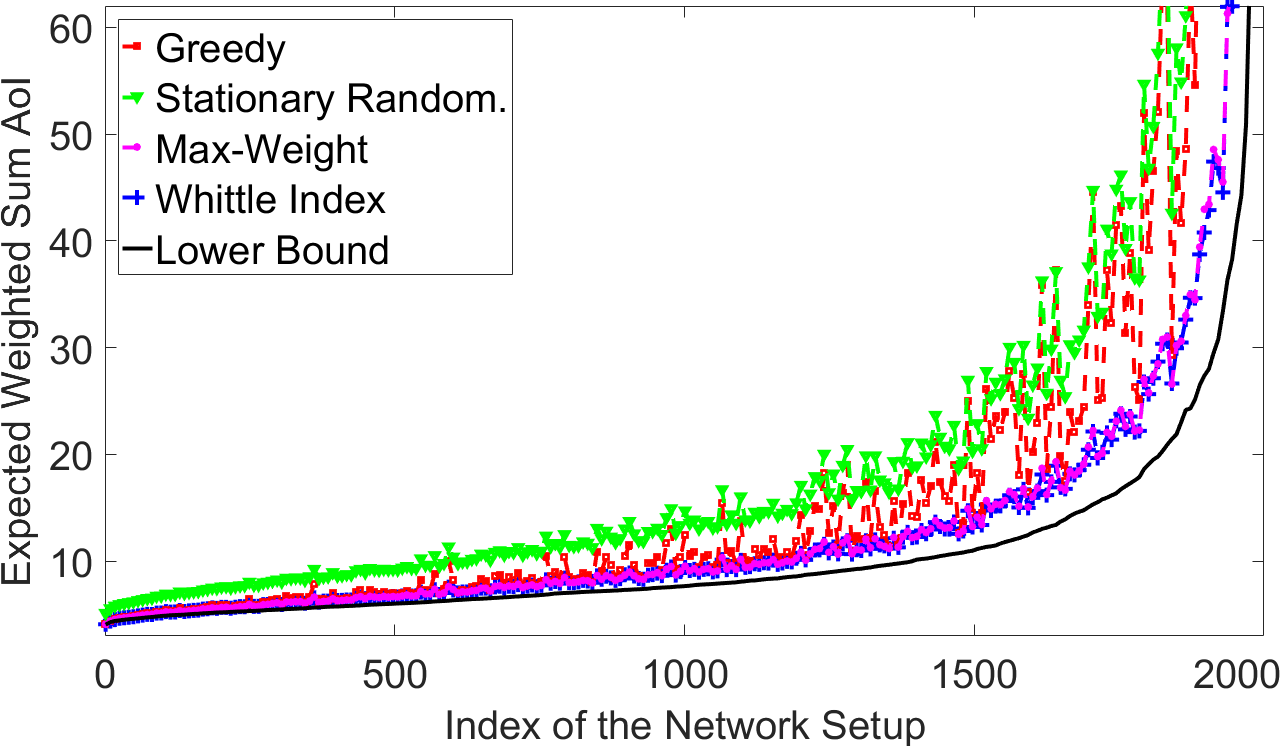}
\end{center}
\caption{Networks with $M=4, T=2, K=50,000, \alpha_i=1, \;\forall i$ and different channel reliabilities $p_i$. For each network setup, the values of $p_i$ are sampled uniformly at random from the range $(0,1)$. For a given network setup, the performance of the policies is an average over $10$ runs. For the sake of clarity, we display $250$ out of the $2,000$ network setups by keeping only every $8$th data point.}\label{fig.Gen_index}
\end{figure}

Our results in Figs.~\ref{fig.Sym} and \ref{fig.Gen} show that the performances of the Max-Weight and Whittle Index policies are comparable to the optimal performance (DP) in every network setting considered. Moreover, the results in Fig.~\ref{fig.Sym} support the optimality of the Greedy, Max-Weight and Whittle Index policies for any symmetric network. Figs.~\ref{fig.Gen}, \ref{fig.Gen_client} and \ref{fig.Gen_index} suggest that, in general, the Max-Weight and Whittle Index Policies outperform Greedy and Randomized. An important feature of all policies examined in this paper is that they require low computational resources even for networks with a large number of clients. 
\section{CONCLUDING REMARKS}\label{sec.Conclusion}
This paper considered a wireless broadcast network with a BS sending time-sensitive information to multiple clients over unreliable channels. We studied the problem of optimizing scheduling decisions with respect to the expected weighted sum AoI of the clients in the network. Our main contributions include developing the Greedy, Randomized, Max-Weight and Whittle's Index policies; showing that for the case of symmetric networks, Greedy, Max-Weight and Whittle's Index are AoI-optimal; and deriving performance guarantees for all four low-complexity policies. Numerical results demonstrate the strong performances of the Max-Weight and Whittle's Index policies in a variety of network conditions.

The mathematical model in Sec.~\ref{sec.Model} describes a network that periodically generates packets at the BS and then transmits those packets to the clients. It is easy to see that the same model can represent other types of networks. A simple example is a polling network in which the BS requests packets from the clients and each client, once polled, generates fresh data and transmits that data back to the BS. This network with uplink traffic and on-demand generation of data can be represented by our model for the case $T=1$. Interesting extensions of this work include considering stochastic arrivals, time-varying channels and multi-hop networks.


\appendices
\section{Proof of Lemma \ref{lem.Greedy}}\label{app.Lemma_Greedy}
\noindent \textbf{Lemma \ref{lem.Greedy}} (Round Robin). Without loss of generality, reorder the client index $i$ in descending order of $\vec{h}_1$, with client $1$ having the highest $h_{1,i}$ and client $M$ the lowest $h_{1,i}$. The Greedy policy \textbf{delivers} packets according to the index sequence $(1,2,\cdots,M,1,2,\cdots)$ until the end of the time-horizon $K$, i.e. Greedy follows a Round Robin pattern.

\begin{proof}
Suppose that $p_i=1$ for all clients, meaning that every transmission is a successful packet delivery. Consider the first frame $k=1$ and assume that there are less clients in the network than slots in a frame, i.e. $M < T$. In this case, the Greedy Policy delivers a packet to client $1$ in the first slot, client $2$ in the second slot, and so on, until the $M$th packet is delivered. At this point, there are no undelivered packets left, and Greedy idles until the end of the frame. In the next frame $k=2$, new packets are generated at the BS and the value of $h_{k,i}$ is updated to $1$ for all clients. Since Greedy breaks ties arbitrarily, we choose to select clients in the same order $(1,2,\cdots,M)$ during frame $k=2$ and during all subsequent frames. This client ordering characterizes a circular order. Thus, for the case $M<T$ and $p_i=1$, the Greedy Policy delivers packets to clients following a Round Robin pattern.

Now, consider the case $M\geq T$ and $p_i=1$. Let $m_1 \in \mathbb{N}$ and $m_2 \in \{0,1,\cdots,T-1\}$ be the quotient and remainder of the division of $M$ by $T$, namely $M=m_1T+m_2$. For simplicity of exposition, let the client index $i$ be \emph{reordered in descending order of $h_{k,i}$ at the beginning of every frame} $k$. Then, within every frame $k$, the Greedy Policy schedules clients in the following order $(1,\cdots,T)$. The evolution of the Greedy Policy is described in detail next:
\begin{itemize}
\item In the first frame, the Greedy policy delivers packets to clients $1$ through $T$ in order. 
\item At the beginning of the second frame, new packets are generated at the BS and the value of $h_{k,i}$ is updated to $1$ for clients $\{1,\cdots,T\}$ and incremented by $1$ for the remaining clients. Then, the client index $i$ is reordered such that vector $\vec{h}_2$ is in descending order. Reordering can be accomplished with a cyclic shift of $T$ elements, in particular, clients $\{1,\cdots,T\}$ become $\{M-T+1,\cdots,M\}$ and clients that did not receive packets in the first frame have their index subtracted by $T$. With these reordered indexes, during the second frame, the Greedy policy delivers packets to clients $1$ through $T$ in order. 
\item Similarly, at the beginning of the third frame, new packets are generated at the BS and the value of $h_{k,i}$ is updated to $1$ for clients $\{1,\cdots,T\}$ and incremented by $1$ for the remaining clients. The vector $\vec{h}_3$ is reordered by applying the same cyclic shift of $T$ elements. Notice that the value of $h_{3,i}$ is $h_{3,i}=1$ for the clients that received packets in the second frame and $h_{3,i}=2$ for the clients that received packets in the first frame. During the third frame, Greedy delivers packets to clients $1$ through $T$ in order. 
\item This process is repeated until frame $k=m_1$. Then, at the beginning of frame $k=m_1+1$, the reordered vector of $h_{k,i}$ is
\begin{equation}
\vec{h}_{m_1+1}=\left[ \begin{array}{c} h_{1,i}+m_1 \\ m_1 \\ m_1-1 \\ \vdots \\ 2 \\ 1 \end{array} \right] \begin{array}{c} m_2\mbox{ elements} \\ T\mbox{ elements} \\ T\mbox{ elements} \\ \vdots \\ T\mbox{ elements} \\ T\mbox{ elements} \end{array}
\end{equation}
Clients $\{1,\cdots,m_2\}$ are the only ones that did not receive a packet so far. During frame $k=m_1+1$, the Greedy Policy delivers packets to clients $1$ through $T$ in order, where, by definition, $T>m_2$. 
\item Therefore, at the beginning of frame $k=m_1+2$, all clients have received at least one packet and the reordered vector of $h_{k,i}$ is
\begin{equation}\label{eq.vectorh_k}
\vec{h}_{m_1+2}=\left[ \begin{array}{c} m_1+1 \\ m_1 \\ m_1-1 \\ \vdots \\ 2 \\ 1 \end{array} \right] \begin{array}{c} m_2\mbox{ elements} \\ T\mbox{ elements} \\ T\mbox{ elements} \\ \vdots \\ T\mbox{ elements} \\ T\mbox{ elements} \end{array}
\end{equation}
During frame $k=m_1+2$, the Greedy policy delivers packets to clients $1$ through $T$ in order.
\item At the beginning of frame $k=m_1+3$, the reordered vector of $h_{k,i}$ is
\begin{equation}\label{eq.vectorh_k2}
\vec{h}_{m_1+3}=\left[ \begin{array}{c} m_1+1 \\ m_1 \\ m_1-1 \\ \vdots \\ 2 \\ 1 \end{array} \right] \begin{array}{c} m_2\mbox{ elements} \\ T\mbox{ elements} \\ T\mbox{ elements} \\ \vdots \\ T\mbox{ elements} \\ T\mbox{ elements} \end{array}
\end{equation}
Observe that \eqref{eq.vectorh_k2} and \eqref{eq.vectorh_k} are identical. Clearly, in all frames that follow, the same sequence of events occur: i) vector $\vec{h}_{k}$ is updated according to \eqref{eq.evolution_h}; ii) vector $\vec{h}_{k}$ is reordered using a circular shift of $T$ elements, resulting in $\vec{h}_{k}$ identical to \eqref{eq.vectorh_k}; and iii) the Greedy Policy delivers clients $1$ through $T$ in order.
\end{itemize}

The description above for both cases $M<T$ and $M\geq T$ shows that when channels are error-free, namely $p_i=1$, and we iteratively apply cyclic shifts of $T$ elements to the client indexes, the Greedy Policy delivers packets to clients $1$ through $T$ in order at every frame $k$. Equivalently, when no cyclic shift is applied, the Greedy Policy delivers packets to clients in circular order. 

When channels are unreliable, the only difference in the analysis is that each packet transmission to client $i$ fails with probability $p_i \in (0,1], \forall i$. According to Remark \ref{rem.Greedy}, in the event of a transmission failure, Greedy continues to transmit to the same client. Thus, transmission failures do not affect the order in which packets are \emph{delivered}. Hence, irrespective of the network setup, the Greedy Policy delivers packets following a Round Robin pattern until the end of the time-horizon. 
\end{proof}
\section{Proof of Lemma \ref{prop.Greedy}}\label{app.Prop_Greedy}
\noindent \textbf{Lemma \ref{prop.Greedy}} (Optimality of Greedy for error-free channels). Consider a symmetric network with error-free channels $p_i=1$ and weights $\alpha_i=\alpha>0, \forall i$. Among the class of admissible policies $\Pi$, the Greedy policy attains the minimum sum AoI \eqref{eq.EWSAoI}, namely
\begin{equation}
J_K^G \leq J_K^\pi, \forall \pi \in \Pi \; . 
\end{equation}

\begin{proof}

Prior to delving into the proof, we introduce some notation. Let $\mathbb{I}_i(.)$ be an indicator function that takes the value $\mathbb{I}_i(s)=1$ if $\{i \notin s \}$ and $\mathbb{I}_i(s)=0$, otherwise. Recall from Sec.~\ref{sec.AoI} that $\hat{s}_k$ represents the set of clients that successfully received packets during frame $k$. Then, we can indicate that client $i$ \emph{did not receive} a packet during frame $k$ using $\mathbb{I}_i(\hat{s}_k)=1$ and the evolution of $h_{k,i}$ in \eqref{eq.evolution_h} can be rewritten as $h_{k+1,i}=h_{k,i}\mathbb{I}_i(\hat{s}_k)+1$. Using vector notation, let $\vec{\mathbb{I}}(\hat{s}_k)=[\mathbb{I}_1(\hat{s}_k)~\mathbb{I}_2(\hat{s}_k)~\cdots~\mathbb{I}_M(\hat{s}_k)]^T$ and denote by $\vec{h}_k \odot \vec{\mathbb{I}}(\hat{s}_k)$ the entrywise product of vectors $\vec{h}_k$ and $\vec{\mathbb{I}}(\hat{s}_k)$. A simple expression for the evolution of $\vec{h}_{k}$ at the beginning of frame $k+1$ is
\begin{equation}\label{eq.evolution_h2}
\vec{h}_{k+1}=\vec{h}_k \odot \vec{\mathbb{I}}(\hat{s}_k)+\vec{1} \; ,
\end{equation}
where $\vec{1}$ is the unity column vector of length $M$.

Next, we use mathematical induction to show that $\vec{h}_{k+1}$ can be expressed as a function of the initial AoI, $\vec{h}_1$, and of the clients that received packets during frames $1$ through $k$, namely $\{\hat{s}_j\}_{j=1}^{k}$, as follows
\begin{equation}\label{eq.evolution_h_hyp}
\vec{h}_{k+1}=\vec{h}_{1} \odot \vec{\mathbb{I}}\left(\bigcup\limits_{j=1}^{k} \hat{s}_{j}\right) +\sum_{a=2}^{k}\vec{\mathbb{I}}\left(\bigcup\limits_{j=a}^{k} \hat{s}_{j}\right)+\vec{1} \; .
\end{equation}

\noindent \textbf{Base case}: substitute $k=1$ in \eqref{eq.evolution_h_hyp}. The expression is identical to \eqref{eq.evolution_h2}.

\noindent \textbf{Inductive step}: Assume that \eqref{eq.evolution_h_hyp} holds for $\vec{h}_k$. The expression of $\vec{h}_{k+1}$ can be obtained by substituting \eqref{eq.evolution_h_hyp} in \eqref{eq.evolution_h2} as follows
\begin{align}
\vec{h}_{k+1}&=\vec{h}_k \odot \vec{\mathbb{I}}(\hat{s}_k)+\vec{1} \nonumber\\
             &=\left[\vec{h}_{1} \odot \vec{\mathbb{I}}\left(\bigcup\limits_{j=1}^{k-1} \hat{s}_{j}\right) +\sum_{a=2}^{k-1}\vec{\mathbb{I}}\left(\bigcup\limits_{j=a}^{k-1} \hat{s}_{j}\right)+\vec{1}\right] \odot \vec{\mathbb{I}}(\hat{s}_k)+\vec{1} \nonumber\\
             &=\vec{h}_{1} \odot \vec{\mathbb{I}}\left(\bigcup\limits_{j=1}^{k} \hat{s}_{j}\right)+\sum_{a=2}^{k-1}\vec{\mathbb{I}}\left(\bigcup\limits_{j=a}^{k} \hat{s}_{j}\right)+\vec{\mathbb{I}}(\hat{s}_k) +\vec{1} \nonumber\\
             &=\vec{h}_{1} \odot \vec{\mathbb{I}}\left(\bigcup\limits_{j=1}^{k} \hat{s}_{j}\right)+\sum_{a=2}^{k}\vec{\mathbb{I}}\left(\bigcup\limits_{j=a}^{k} \hat{s}_{j}\right)+\vec{1} \; ,
\end{align}
which is identical to the expression in \eqref{eq.evolution_h_hyp}. The induction is complete.

Expression \eqref{eq.evolution_h_hyp} is central for this proof. To show that for a symmetric network with error-free channels, the Greedy policy minimizes $J_K^\pi$, i.e.
\begin{equation}
J_K^G \leq J_K^\pi, \forall \pi \in \Pi \quad \mbox{where } J_K^{\pi}=\frac{\alpha}{KM}\sum_{k=1}^{K}\sum_{i=1}^{M} h_{k,i} \; ,
\end{equation}
it suffices to show that employing Greedy yields the lowest sum $\sum_{i=1}^{M} h_{k,i}$ in every frame $k \in \{1,2,\cdots,K\}$. 
According to \eqref{eq.evolution_h_hyp}, the sum of the elements of $\vec{h}_k$ can be expressed as
\begin{align}
\sum_{i=1}^{M} h_{k,i}=&\sum_{i=1}^M \left[ h_{1,i}\mathbb{I}_i\left(\bigcup\limits_{j=1}^{k-1} \hat{s}_{j}\right)+\sum_{a=2}^{k-1}\mathbb{I}_i\left(\bigcup\limits_{j=a}^{k-1} \hat{s}_{j}\right) +1 \right] \nonumber\\
=&\sum_{i=1}^M h_{1,i}\mathbb{I}_i\left(\bigcup\limits_{j=1}^{k-1} \hat{s}_{j}\right)+\sum_{a=2}^{k-1}\sum_{i=1}^M\mathbb{I}_i\left(\bigcup\limits_{j=a}^{k-1} \hat{s}_{j}\right) + M \; . \label{eq.sum_h}
\end{align}

In the remaining part of this proof, we use Lemma \ref{lem.Greedy} (Circular Order) to show that the Greedy policy minimizes \eqref{eq.sum_h} for every frame $k$. Without loss of generality, we reorder the client index $i$ in descending order of $\vec{h}_1$. Then, Lemma \ref{lem.Greedy} states that Greedy delivers packets to clients in circular order, following the index sequence $(1,2,\cdots,M,1,2,\cdots)$. Clearly, when Greedy is employed in a network with $p_i=1,\forall i$, the following properties hold:
\begin{enumerate}[(i)]
\item the number of packets delivered in any frame $k$ is $|\hat{s}_k^{\;G}|=\min(M,T), \forall k \in \{1,2,\cdots,K\}$;
\item the set of clients that receive at least one packet during the first $k$ frames is 
\begin{equation}
\bigcup\limits_{j=1}^{k} \hat{s}_{j}^{\;G}=\{1,2,\cdots,\min(M,kT)\} \; ;
\end{equation}
\item by minimizing the number of common elements in neighboring sets, the circular order of Greedy maximizes the number of clients that receive at least one packet during the first $k$ frames. Thus, for every $\pi \in \Pi$ we have
\begin{equation}
\left|\bigcup\limits_{j=1}^{k} \hat{s}_{j}^{\;\pi}\right| \leq \left|\bigcup\limits_{j=1}^{k} \hat{s}_{j}^{\;G}\right|=\min(M,kT) \; ;
\end{equation}
\item for the same reason, for every $\pi \in \Pi$ and $a \in \{1,2,\cdots,k\}$ we also have
\begin{equation}
\left|\bigcup\limits_{j=a}^{k} \hat{s}_{j}^{\;\pi}\right| \leq \left|\bigcup\limits_{j=a}^{k} \hat{s}_{j}^{\; G}\right| = \min\{M,(k-a+1)T\} \; .
\end{equation}
\end{enumerate}
Properties (ii) and (iv) are used to show that Greedy minimizes the RHS of \eqref{eq.sum_h}.

The first term in the RHS of \eqref{eq.sum_h} is the sum of the elements of $\vec{h}_1$ that are associated with clients that did not receive packets during frames $1$ through $k-1$. Property (ii) shows that the Greedy policy minimizes this term by delivering packets to the clients with highest value of $h_{1,i}$, namely $\{1,2,\cdots,\min(M,(k-1)T)\}$. 

The second term in the RHS of \eqref{eq.sum_h} is a double sum. By the definition of $\mathbb{I}_i(.)$, it follows that
\begin{equation}
\sum_{a=2}^{k-1}\sum_{i=1}^M\mathbb{I}_i\left(\bigcup\limits_{j=a}^{k-1} \hat{s}_{j}\right) = \sum_{a=2}^{k-1} \left[ M - \left| \bigcup\limits_{j=a}^{k-1} \hat{s}_{j} \right|\right] \; .
\end{equation}
Property (iv) shows that the Greedy policy minimizes this double sum for every value of $a$. Since the last term in the RHS of \eqref{eq.sum_h} is a constant, we conclude that the Greedy policy minimizes the sum \eqref{eq.sum_h} in every frame $k$ and, consequently, the value of the objective function, $J_K^\pi$.\end{proof}
\section*{Acknowledgment}
This work was supported by NSF Grants AST-1547331, CNS-1713725, and CNS-1701964, by a grant from the army research office (ARO), by METU and by CAPES/Brazil.

\ifCLASSOPTIONcaptionsoff
  \newpage
\fi

\bibliographystyle{IEEEtran}
\bibliography{IEEEabrv}

\begin{thebibliography}{10}
\providecommand{\url}[1]{#1}
\csname url@samestyle\endcsname
\providecommand{\newblock}{\relax}
\providecommand{\bibinfo}[2]{#2}
\providecommand{\BIBentrySTDinterwordspacing}{\spaceskip=0pt\relax}
\providecommand{\BIBentryALTinterwordstretchfactor}{4}
\providecommand{\BIBentryALTinterwordspacing}{\spaceskip=\fontdimen2\font plus
\BIBentryALTinterwordstretchfactor\fontdimen3\font minus
  \fontdimen4\font\relax}
\providecommand{\BIBforeignlanguage}[2]{{%
\expandafter\ifx\csname l@#1\endcsname\relax
\typeout{** WARNING: IEEEtran.bst: No hyphenation pattern has been}%
\typeout{** loaded for the language `#1'. Using the pattern for}%
\typeout{** the default language instead.}%
\else
\language=\csname l@#1\endcsname
\fi
#2}}
\providecommand{\BIBdecl}{\relax}
\BIBdecl

\bibitem{AoI_broadcast}
I.~Kadota, E.~Uysal-Biyikoglu, R.~Singh, and E.~Modiano, ``Minimizing the age
  of information in broadcast wireless networks,'' in \emph{Proceedings of IEEE
  Allerton}, 2016, pp. 1143--1149.

\bibitem{AoI_update}
S.~Kaul, R.~Yates, and M.~Gruteser, ``Real-time status: How often should one
  update?'' in \emph{Proceedings of IEEE INFOCOM}, 2012, p. 2731–2735.

\bibitem{AoI_multiple}
R.~D. Yates and S.~Kaul, ``Real-time status updating: Multiple sources,'' in
  \emph{Proceedings of IEEE ISIT}, 2012.

\bibitem{AoI_MG1}
L.~Huang and E.~Modiano, ``Optimizing age-of-information in a multi-class
  queueing system,'' in \emph{Proceedings of IEEE ISIT}, 2015.

\bibitem{AoI_management}
M.~Costa, M.~Codreanu, and A.~Ephremides, ``On the age of information in status
  update systems with packet management,'' \emph{IEEE Transactions on
  Information Theory}, vol.~62, no.~4, pp. 1897--1910, 2016.

\bibitem{AoI_path}
C.~Kam, S.~Kompella, G.~D. Nguyen, and A.~Ephremides, ``Effect of message
  transmission path diversity on status age,'' \emph{IEEE Transactions on
  Information Theory}, vol.~62, pp. 1360--1374, 2016.

\bibitem{AoI_errors}
K.~Chen and L.~Huang, ``Age-of-information in the presence of error,'' in
  \emph{Proceedings of IEEE ISIT}, 2016, p. 2579–2583.

\bibitem{AoI_gamma}
E.~Najm and R.~Nasser, ``Age of information: The gamma awakening,'' in
  \emph{Proceedings of IEEE ISIT}, 2016, pp. 2574--2578.

\bibitem{AoI_nonlinear}
A.~Kosta, N.~Pappas, A.~Ephremides, and V.~Angelakis, ``Age and value of
  information: Non-linear age case,'' in \emph{Proceedings of IEEE ISIT}, 2017.

\bibitem{AoI_energy15}
B.~T. Bacinoglu, E.~T. Ceran, and E.~Uysal-Biyikoglu, ``Age of information
  under energy replenishment constraints,'' in \emph{Proceedings of IEEE ITA},
  2015.

\bibitem{AoI_energy17}
B.~T. Bacinoglu and E.~Uysal-Biyikoglu, ``Scheduling status updates to minimize
  age of information with an energy harvesting sensor,'' in \emph{Proceedings
  of IEEE ISIT}, 2017.

\bibitem{AoI_lazy}
R.~D. Yates, ``Lazy is timely: Status updates by an energy harvesting source,''
  in \emph{Proceedings of IEEE ISIT}, June 2015, pp. 3008--3012.

\bibitem{UpdateOrWait17}
Y.~Sun, E.~Uysal-Biyikoglu, R.~Yates, C.~E. Koksal, and N.~B. Shroff, ``Update
  or wait: How to keep your data fresh,'' \emph{IEEE Transactions on
  Information Theory}, 2017.

\bibitem{AoI_scheduling}
Q.~He, D.~Yuan, and A.~Ephremides, ``Optimizing freshness of information: On
  minimum age link scheduling in wireless systems,'' in \emph{Proceedings of
  IEEE WiOpt}, 2016.

\bibitem{PAoI_scheduling}
------, ``On optimal link scheduling with min-max peak age of information in
  wireless systems,'' in \emph{Proceedings of IEEE ICC}, 2016.

\bibitem{AoI_cache}
R.~D. Yates, P.~Ciblat, A.~Yener, and M.~Wigger, ``Age-optimal constrained
  cache updating,'' in \emph{Proceedings of IEEE ISIT}, 2017.

\bibitem{AoI_multiaccess}
S.~Kaul and R.~D. Yates, ``Status updates over unreliable multiaccess
  channels,'' in \emph{Proceedings of IEEE ISIT}, 2017.

\bibitem{AoI_sync}
C.~Joo and A.~Eryilmaz, ``Wireless scheduling for information freshness and
  synchrony: Drift-based design and heavy-traffic analysis,'' in
  \emph{Proceedings of IEEE WiOpt}, 2017.

\bibitem{AoI_design}
Y.-P. Hsu, E.~Modiano, and L.~Duan, ``Age of information: Design and analysis
  of optimal scheduling algorithms,'' in \emph{Proceedings of IEEE ISIT}, 2017.

\bibitem{AoI_LGFS16}
A.~M. Bedewy, Y.~Sun, , and N.~B. Shroff, ``Optimizing data freshness,
  throughput, and delay in multi-server information-update systems,'' in
  \emph{Proceedings of IEEE ISIT}, 2016.

\bibitem{AoI_LGFS17}
------, ``Age-optimal information updates in multihop networks,'' in
  \emph{Proceedings of IEEE ISIT}, 2017.

\bibitem{AoI_VANET}
S.~Kaul, M.~Gruteser, V.~Rai, and J.~Kenney, ``Minimizing age of information in
  vehicular networks,'' in \emph{Proceedings of IEEE SECON}, 2011, pp.
  350--358.

\bibitem{AoI_buffer}
C.~Kam, S.~Kompella, G.~D. Nguyen, J.~E. Wieselthier, and A.~Ephremides,
  ``Controlling the age of information: Buffer size, deadline, and packet
  replacement,'' in \emph{Proceedings of IEEE MILCOM}, 2016, pp. 301--306.

\bibitem{AoI_emulation}
C.~Kam, S.~Kompella, and A.~Ephremides, ``Experimental evaluation of the age of
  information via emulation,'' in \emph{Proceedings of IEEE MILCOM}, 2015, pp.
  1070--1075.

\bibitem{AoI_LUPMAC}
A.~Franco, E.~Fitzgerald, B.~Landfeldt, N.~Pappas, and V.~Angelakis, ``Lupmac:
  A cross-layer mac technique to improve the age of information over dense
  wlans,'' in \emph{Proceedings of IEEE ICT}, 2016.

\bibitem{single_link}
P.~P. Bhattacharya and A.~Ephremides, ``Optimal scheduling with strict
  deadlines,'' \emph{IEEE Transactions on Automatic Control}, vol.~34, pp.
  721--728, July 1989.

\bibitem{theoryofQoS}
I.-H. Hou, V.~Borkar, and P.~R. Kumar, ``A theory of qos for wireless,'' in
  \emph{Proceedings of IEEE INFOCOM}, Apr. 2009, pp. 486--494.

\bibitem{delay}
K.~S. Kim, C.-P. Li, I.~Kadota, and E.~Modiano, ``Optimal scheduling of
  real-time traffic in wireless networks with delayed feedback,'' in
  \emph{Proceedings of IEEE Allerton Conference on Communication, Control and
  Computing}, Oct. 2015, pp. 1143--1149.

\bibitem{index_schedule}
V.~Raghunathan, V.~Borkar, M.~Cao, and P.~R. Kumar, ``Index policies for
  real-time multicast scheduling for wireless broadcast systems,'' in
  \emph{Proceedings of IEEE INFOCOM}, Apr. 2008.

\bibitem{TSLS_15}
B.~Li, R.~Li, and A.~Eryilmaz, ``Throughput-optimal wireless scheduling with
  regulated inter-service times,'' \emph{IEEE/ACM Transactions on Networking},
  vol.~23, pp. 1542--1552, 2015.

\bibitem{TSLS_16}
------, ``Wireless scheduling design for optimizing both service regularity and
  mean delay in heavy-traffic regimes,'' \emph{IEEE/ACM Transactions on
  Networking}, vol.~24, pp. 1867--1880, 2016.

\bibitem{index_regularity}
R.~Singh, X.~Guo, and P.~Kumar, ``Index policies for optimal mean-variance
  trade-off of inter-delivery times in real-time sensor networks,'' in
  \emph{Proceedings of IEEE INFOCOM}, Jan. 2015, pp. 505--512.

\bibitem{regularity_reliable}
X.~Guo, R.~Singh, P.~Kumar, and Z.~Niu, ``A high reliability asymptotic
  approach for packet inter-delivery time optimization in cyber-physical
  systems,'' in \emph{Proceedings of ACM International Symposium on Mobile Ad
  Hoc Networking and Computing}, 2015, pp. 197--206.

\bibitem{regularity_smooth}
R.~Singh and A.~Stolyar, ``Maxweight scheduling: ``smoothness'' of the service
  process,'' in \emph{Proceedings of IEEE INFOCOM}, 2016.

\bibitem{regularity_frequency}
X.~Zheng, Z.~Cai, J.~Li, and H.~Gao, ``Scheduling flows with multiple service
  frequency constraints,'' \emph{IEEE Internet of Things Journal}, vol.~4, pp.
  496--504, 2017.

\bibitem{regularity_round_robin}
B.~Li, A.~Eryilmaz, and R.~Srikant, ``Emulating round-robin in wireless
  networks,'' in \emph{Proceedings of ACM MobiHoc}, 2017.

\bibitem{coupling}
D.~Stoyan, \emph{Comparison Methods for Queues and Other Stochastic Models},
  ser. Wiley Series in Probability and Statistics.\hskip 1em plus 0.5em minus
  0.4em\relax Wiley, 1983.

\bibitem{dynProg}
D.~Bertsekas, \emph{Dynamic Programming and Optimal Control}, 3rd~ed.\hskip 1em
  plus 0.5em minus 0.4em\relax Athena Scientific, 2005, vol.~1.

\bibitem{DSP}
R.~G. Gallager, \emph{Stochastic Processes: Theory for Applications}.\hskip 1em
  plus 0.5em minus 0.4em\relax Cambridge University Press, 2013.

\bibitem{lyapunov}
M.~J. Neely, \emph{Stochastic Network Optimization with Application to
  Communication and Queueing Systems}.\hskip 1em plus 0.5em minus 0.4em\relax
  Morgan and Claypool Publishers, 2010.

\bibitem{multicast}
K.~S. Kim, C.-P. Li, and E.~Modiano, ``Scheduling multicast traffic with
  deadlines in wireless networks,'' in \emph{Proceedings of IEEE INFOCOM}, May
  2014, pp. 2193--2201.

\bibitem{RMAB}
P.~Whittle, ``Restless bandits: Activity allocation in a changing world,''
  \emph{Journal of Applied Probability}, vol.~25, pp. 287--298, 1988.

\bibitem{index_myopic}
P.~Mansourifard, T.~Javidi, and B.~Krishnamachari, ``Optimality of myopic
  policy for a class of monotone affine restless multi-armed bandits,'' in
  \emph{Proceedings of IEEE CDC}, Dec. 2012, pp. 877--882.

\bibitem{index_multichannelaccess}
K.~Liu and Q.~Zhao, ``Indexability of restless bandit problems and optimality
  of whittle index for dynamic multichannel access,'' \emph{IEEE Transactions
  on Information Theory}, vol.~56, pp. 5547--5567, Nov. 2010.

\bibitem{index_assympt}
R.~R. Weber and G.~Weiss, ``On an index policy for restless bandits,''
  \emph{Journal of Applied Probability}, vol.~27, no.~3, pp. 637--648, 1990.

\bibitem{RMAB_book}
J.~Gittins, K.~Glazebrook, and R.~Weber, \emph{Multi-armed Bandit Allocation
  Indices}, 2nd~ed.\hskip 1em plus 0.5em minus 0.4em\relax Wiley, Mar. 2011.

\bibitem{scheduling_coupling}
A.~Ganti, E.~Modiano, and J.~N. Tsitsiklis, ``Optimal transmission scheduling
  in symmetric communication models with intermittent connectivity,''
  \emph{IEEE Transactions on Information Theory}, vol.~53, no.~3, Mar. 2007.

\bibitem{inequality}
Y.-C. Li and C.-C. Yeh, ``Some equivalent forms of bernoulli's inequality: A
  survey,'' \emph{Applied Mathematics}, vol.~4, no.~7, July 2013.

\end{thebibliography}

\newpage
\setcounter{page}{1}
\twocolumn[
\begin{@twocolumnfalse}
{
\begin{center}
\fontsize{18pt}{25pt}\selectfont Supplementary Material for the paper ``Scheduling Policies for Minimizing Age of Information in Broadcast Wireless Networks''
\end{center}
\vspace{1cm}
}
\end{@twocolumnfalse}]

\section{Proof of Theorem \ref{theo.Greedy}}\label{app.Theo_Greedy}
\noindent \textbf{Theorem \ref{theo.Greedy}} (Optimality of Greedy). Consider a symmetric network with channel reliabilities $p_i=p \in (0,1]$ and weights $\alpha_i=\alpha > 0, \forall i$. Among the class of admissible policies $\Pi$, the Greedy Policy attains the minimum expected sum AoI \eqref{eq.EWSAoI}, namely $G=\argmin_{\pi \in \Pi}\mathbb{E}\left[J_K^{\pi}\right]$.

\begin{proof}
To show that the Greedy Policy minimizes the EWSAoI in \eqref{eq.EWSAoI} for symmetric networks, we utilize a stochastic dominance argument to compare the evolution of $\vec{h}_k$ when Greedy is employed to that when an arbitrary policy $\pi$ is employed. For the sake of simplicity and without loss of optimality, in this proof we assume that $\pi$ is work-conserving. There is no loss of optimality since for every non work-conserving policy, there is at least one work-conserving policy that is strictly dominant.

Let $SH^\pi_k$ be the random variable that represents the sum of the elements of $\vec{h}_k$ when $\pi$ is employed. Using this notation and the symmetry assumptions of Theorem \ref{theo.Greedy}, the objective function in \eqref{eq.Objective} becomes
\begin{align}
\min_{\pi \in \Pi}\mathbb{E}\left[J_K^{\pi}\right] = & \frac{1}{KM}\min_{\pi \in \Pi}\mathbb{E}\left[\sum_{k=1}^{K}\sum_{i=1}^{M} \alpha \; h_{k,i}\right] \nonumber \\
= & \frac{\alpha}{KM}\min_{\pi \in \Pi}\mathbb{E}\left[\sum_{k=1}^{K}SH_k^\pi\right] \; .\label{eq.DP1}
\end{align}

For introducing the concept of stochastic dominance, denote the stochastic process associated with the sequence $\{SH^\pi_k\}_{k=1}^{K}$ as $SH^\pi$ and its sample path as $sh^{\pi}$. Let $\mathbb{D}$ be the space of all sample paths $sh^{\pi}$. Define by $\mathcal{F}$ the set of measurable functions $f:\mathbb{D} \rightarrow \mathbb{R}^+$ such that $f(sh^{G}) \leq f(sh^\pi)$ for every $sh^{G} , sh^\pi \in \mathbb{D}$ which satisfy $sh^{G}_k \leq sh^\pi_k, \forall k$. 

\begin{definition}
(Stochastic Dominance) We say that $SH^{G}$ is stochastically smaller than $SH^\pi$ and write $SH^{G} \leq_{st} SH^\pi$ if $\;P\{f(SH^{G})>z\}\leq P\{f(SH^\pi)>z\}, \forall z \in \mathbb{R}, \forall f \in \mathcal{F}$. 
\end{definition}

Since $f(SH^\pi)$ is positive valued, $SH^{G} \leq_{st} SH^\pi$ implies\footnote{Recall that for any positive valued $X$, it follows that $\mathbb{E}[X]=\int_{x=0}^{\infty}(1-P\{X \leq x\})dx=\int_{x=0}^{\infty}P\{X > x\}dx$.} $\mathbb{E}[f(SH^{G})] \leq \mathbb{E}[f(SH^\pi)], \forall f \in \mathcal{F}$. Knowing that one function that satisfies the conditions in $\mathcal{F}$ is $f(SH^\pi)=\sum_{k=1}^K SH^\pi_k$, it follows that if $SH^{G} \leq_{st} SH^\pi, \forall \pi \in \Pi$, then $\mathbb{E}[\sum_{k=1}^K SH^{G}_k] \leq \mathbb{E}[\sum_{k=1}^K SH^\pi_k], \forall \pi \in \Pi$, which is our target expression in \eqref{eq.DP1}. Therefore, it follows that for establishing the optimality of G, it is sufficient to confirm that $SH^{G}$ is stochastically smaller than $SH^{\pi}, \forall \pi \in \Pi$. 

Stochastic dominance can be demonstrated using its definition directly. However, this is often complex for it involves comparing the probability distributions of $SH^{G}$ and $SH^\pi$. Instead, we use the following result from \cite{coupling}, which is also used in works such as \cite{single_link,scheduling_coupling,index_schedule}: for verifying that $SH^{G} \leq_{st} SH^\pi$, it is sufficient to show that there exists two stochastic processes $\widetilde{SH}^{G}$ and $\widetilde{SH}^\pi$ such that
\begin{enumerate}
\item[(i)] $SH^\pi$ and $\widetilde{SH}^\pi$ have the same probability distribution;
\item[(ii)] $\widetilde{SH}^{G}$ and $\widetilde{SH}^\pi$ are on a common probability space;
\item[(iii)] $SH^{G}$ and $\widetilde{SH}^{G}$ have the same probability distribution;
\item[(iv)] $\widetilde{SH}^{G}_k \leq \widetilde{SH}^\pi_k$, with probability $1$, $\forall k$.
\end{enumerate}
This result allows us to establish stochastic dominance between $SH^{G}$ and $SH^\pi$ by properly designing the auxiliary processes $\widetilde{SH}^{G}$ and $\widetilde{SH}^\pi$. This design is achieved by utilizing Stochastic Coupling.

Prior to discussing stochastic coupling, we introduce the channel state. Let $E_i(k,n) \sim Ber(p)$ be the random variable that represents the channel state of client $i$ during slot $(k,n)$
\begin{equation}\label{eq.channel}
E_i(k,n)= \left\{ \begin{array}{cll} 
1, & \mbox{w.p. } p   \quad &\mbox{[\emph{Channel ON}]} \; ; \\
0, & \mbox{w.p. } 1-p \quad &\mbox{[\emph{Channel OFF}]} \; . \end{array} \right.
\end{equation}
The channel state of each client is independent of the channel state of other clients and of scheduling decisions. Note that the BS has no knowledge of the channel state of the clients before transmissions.

\emph{Stochastic coupling} is a method utilized for comparing stochastic processes by imposing a common underlying probability space. We use stochastic coupling to construct $\widetilde{SH}^{\pi}$ and $\widetilde{SH}^{G}$ based on $SH^{\pi}$ and $SH^{G}$, respectively.

Let the process $\widetilde{SH}^{\pi}$ be identical to $SH^{\pi}$. Their (common) probability space is associated with the channel state of the client selected in each slot by policy $\pi$. Now, let us construct $\widetilde{SH}^{G}$ on the same probability space as $\widetilde{SH}^{\pi}$. For that, we couple $\widetilde{SH}^{G}$ to $\widetilde{SH}^{\pi}$ by dynamically connecting the channel state of Greedy to the channel state of policy $\pi$ as follows. Suppose that in slot $(k,n)$, policy $\pi$ schedules client $j$ while Greedy schedules client $i$, then, for the duration of that slot, we assign $E_i(k,n)\leftarrow E_j(k,n)$. For example, if the outcome associated with policy $\pi$ is $E_j(k,n)=1$, then we impose that Greedy has $E_i(k,n)=1$, regardless of the client $i$ selected by Greedy. This dynamic assignment imposes that, at every slot, the channel state of Greedy is identical to the channel state of $\pi$. Notice that this is only possible because the channel state $E_i(k,n)$ is i.i.d. with respect to the clients and slots, which is the same reason for $\widetilde{SH}^{G}$ and $SH^G$ having the same probability distribution. 

Returning to our four conditions, it follows from the coupling method described above that (i), (ii) and (iii) are satisfied. Thus, the only condition that remains to be shown is
\begin{equation}\label{eq.condition_final}
\mbox{(iv)} \;\; \widetilde{SH}^{G}_k \leq \widetilde{SH}^\pi_k, \mbox{with probability 1, }\forall k.
\end{equation}

Coupling between $\widetilde{SH}^{\pi}$ and $\widetilde{SH}^{G}$ is the key property to establish (iv). Assume that policy $\pi$ is employed and consider a sample path $\widetilde{sh}^\pi$ spanning the entire time-horizon. Use the sequence of channel states from $\widetilde{sh}^\pi$ to create the coupled sample path $\widetilde{sh}^{G}$. Figure~\ref{fig.SH} illustrates both sample paths. Notice that the scheduling decisions taken during slots in which the channel state is OFF cannot change the relationship ($\leq$ or $\geq$) between $\widetilde{sh}^{\pi}_k$ and $\widetilde{sh}^{G}_k$. Since these slots are irrelevant for comparing $\widetilde{sh}^{\pi}_k$ and $\widetilde{sh}^{G}_k$, they can be removed from the analysis and we can focus on slots with error-free channels.

\begin{figure}[b]
\begin{center}
\includegraphics[height=4.7cm]{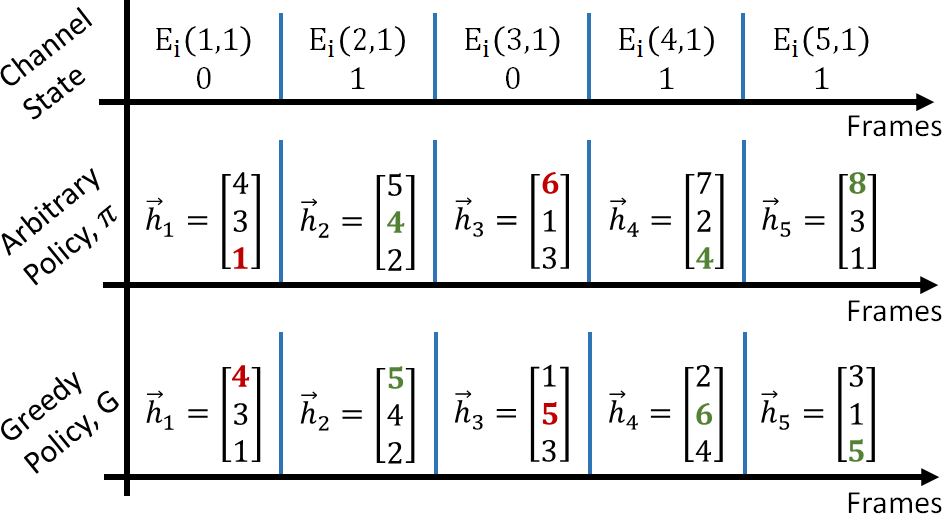}
\end{center}
\caption{Evolution of $\widetilde{sh}^{\pi}$ and $\widetilde{sh}^{G}$ for a network with $M=3$ clients,  $T=1$ slots in a frame, $K=5$ frames, initial AoI $\vec{h}_1=[4,3,1]^T$ and unreliable channels. Recall that channels are ON when $E_i(k,n)=1$ and OFF when $E_i(k,n)=0$. Successful deliveries are represented in green and failed transmissions in red. On the top, channel states associated with the scheduling decisions of the arbitrary policy $\pi$. Notice that, due to coupling, the Greedy policy has the same channel states. On the middle, the evolution of $\vec{h}_k$ when policy $\pi$ is employed. On the bottom, the evolution of $\vec{h}_k$ when Greedy is employed. Comparing the sum of $\vec{h}_k$ over time for both policies, we have $\widetilde{sh}^{\pi}=\{8;11;10;13;12\}$ and $\widetilde{sh}^{G}=\{8;11;9;12;9\}$, and we see that $\widetilde{sh}^{G}_k \leq \widetilde{sh}^\pi_k, \forall k$.}\label{fig.SH}
\end{figure}

Lemma \ref{prop.Greedy} established that, in a network with \emph{error-free channels}, we have $\widetilde{sh}^{G}_k \leq \widetilde{sh}^{\pi}_k$, for every frame $k$ and for every policy $\pi \in \Pi$. The difference between the setup in Appendix \ref{app.Prop_Greedy} and the setup here is that, by removing the slots in which the channel state is OFF, we create frames with different number of slots. However, it is easy to see that the proof in Appendix \ref{app.Prop_Greedy} does not rely on the fact that all frames have the same length. Hence, it follows that $\widetilde{sh}^{G}_k \leq \widetilde{sh}^{\pi}_k$, for every $k$ and for every sample path. Thus establishing condition (iv) and the stochastic dominance argument.
\end{proof}
\section{Proof of Theorem \ref{theo.performance_Greedy}}\label{app.Theo_performance_Greedy}
\noindent \textbf{Theorem \ref{theo.performance_Greedy}} (Performance of Greedy). Consider a network $(M,T,p_i,\alpha_i)$ with an infinite time-horizon. The Greedy policy is $\rho^G$-optimal as $M\rightarrow\infty$, where
\begin{equation}
\rho^G=\frac{\displaystyle\left(\sum_{i=1}^M\alpha_i\right)\left(\sum_{i=1}^M \frac{1}{p_i}\right)\left[1+\frac{C_V^2}{M}\right]+T\left(\sum_{i=1}^M\alpha_i\right)}{\displaystyle\left(\sum_{i=1}^M \sqrt{\frac{\alpha_i}{p_i}}\right)^2+T\left(\sum_{i=1}^M\alpha_i\right)} \; .
\end{equation}

\begin{proof}
The performance guarantee is defined as $\rho^G=U_B^G/L_B$, where the denominator is the universal lower bound in \eqref{eq.LowerBound} and the numerator is an upper bound to the objective function, namely $\lim_{K \rightarrow \infty} \mathbb{E}[J_K^G] \leq U_B^G$, which is derived in this appendix. 

To analyze the evolution of $h_{k,i}$ when the Greedy policy is employed, we utilize the properties introduced in Sec.~\ref{sec.Symmetric}. Without loss of generality, assume in this appendix that the client index $i$ is in descending order of $\vec{h}_1$, as in Lemma~\ref{lem.Greedy}. Then, the properties introduced in Sec.~\ref{sec.Symmetric} can be summarized as follows: i) Greedy \emph{transmits} packets to the same client $i$, uninterruptedly, until a packet is delivered; ii) Greedy \emph{delivers} packets to clients following a Round Robin pattern $(1,2,\cdots,M,1,2,\cdots)$ until the end of the time-horizon; iii) Greedy idles only when all $M$ clients receive their packets in the same frame.

Based on property (i), define $X_i[m]$ as the number of successive transmission attempts to client $i$ that precede the $m$th packet delivery to the same client, with $X_i[0]=0, \forall i$. For a given $i$, the random variables $X_i[m]$ are i.i.d. with geometric distribution. Moreover, transmissions to different clients are independent. Hence, we have
\begin{align}
\mathbb{E}[X_i[m]]&=1/p_i \; ;\\
\mathbb{E}[X_i[m]X_j[m-1]]&=1/p_i p_j \; ;\\
\mathbb{E}\left[X_i^2[m]\right]&=(2-p_i)/p_i^2 \; .
\end{align}

According to property (ii), packets are delivered to clients following a Round Robin pattern. Thus, the total number of packet transmissions (to any client) in the interval between the $(m-1)$th and $m$th deliveries \emph{to client} $i$ is given by 
\begin{align}
B_i[m]=&\sum_{j=i+1}^MX_j[m-1]+\sum_{j=1}^iX_j[m] \; ,
\end{align}
with first moment
\begin{equation}\label{eq.Bi_mean}
\mathbb{E}[B_i[m]]=\sum_{i=1}^M \frac{1}{p_i} \; ,
\end{equation}
and second moment
\begin{align}\label{eq.Bi_power}
\mathbb{E}\left[B_i^2[m]\right]=\sum_{j=1}^M \frac{2-p_j}{p_j^2}+2\sum_{j=1}^M\sum_{k=j+1}^M \frac{1}{p_j p_k} \; .
\end{align}

In addition to packet transmissions, the interval between the $(m-1)$th and $m$th deliveries to client $i$ may also have idle slots. Let $W_i[m]$ be the number of idle slots in this interval. It follows from property (iii) that, if there are idle slots between two consecutive deliveries to client $i$, they occur one after the other and at the end of a frame in which all $M$ packets were delivered, implying that $0 \leq W_i[m] < T$.

The total number of \emph{slots} in the interval between the $(m-1)$th and $m$th packet deliveries to client $i$ is given by $B_i[m]+W_i[m]$. Recall that $I_i[m]$ is defined as the number of \emph{frames} in that interval, hence 
\begin{equation}\label{eq.transmission_cycle}
I_i[m]= \left\lfloor \frac{B_i[m]+W_i[m]}{T} \right\rfloor \; .
\end{equation}
It is evident from the analysis above that, when Greedy is employed, the sequence of packet deliveries to client $i$ is a renewal process with i.i.d. inter-delivery times $I_i[m]$. Therefore, using the generalization of the elementary renewal theorem for renewal-reward processes \cite[Sec.~5.7]{DSP}, we have
\begin{equation}\label{eq.UB_Greedy1}
\lim_{K\rightarrow\infty}\frac{1}{K}\sum_{k=1}^K\mathbb{E}[h_{k,i}]=\frac{\mathbb{E}[I_i[m]^2]}{2\mathbb{E}[I_i[m]]}+\frac{1}{2} \; .
\end{equation}



Next, we find an upper bound to \eqref{eq.UB_Greedy1} that is used for deriving the expression of $U_B^G$. In particular, we combine \eqref{eq.transmission_cycle}, \eqref{eq.Bi_mean} and the fact that $W_i[m]\geq 0$ to obtain a lower bound to the first moment of $I_i[m]$, as follows
\begin{align}\label{eq.SM_upper}
\mathbb{E}[I_i[m]] \geq \mathbb{E}\left[\frac{B_i[m]+W_i[m]}{T}-1\right] \geq \frac{1}{T}\sum_{i=1}^M\frac{1}{p_i} -1\; .
\end{align}
Moreover, an upper bound to the second moment of $I_i^2[m]$ is obtained by using \eqref{eq.transmission_cycle}, \eqref{eq.Bi_power} and the fact that $W_i[m]<T$, as follows 
\begin{align}\label{eq.SM2_lower}
\mathbb{E}&\left[I_i^2[m]\right] \leq \mathbb{E}\left[\left(\frac{B_i[m]+W_i[m]}{T}\right)^2\right] \\
	   &\leq \mathbb{E}\left[\frac{B_i^2[m]}{T^2}+\frac{2B_i[m]}{T}+1\right] \nonumber\\
	   &= \frac{1}{T^2}\left[\sum_{j=1}^M \frac{2-p_j}{p_j^2}+2\sum_{j=1}^M\sum_{k=j+1}^M \frac{1}{p_j p_k}\right]+\frac{2}{T}\left[\sum_{i=1}^M \frac{1}{p_i}\right]+1 \; . \nonumber
\end{align} 
Notice that \eqref{eq.SM_upper} and \eqref{eq.SM2_lower} do not depend on indexes $i$ or $m$.

Finally, substituting \eqref{eq.SM_upper} and \eqref{eq.SM2_lower} into \eqref{eq.UB_Greedy1} and then combining the result with the objective function in \eqref{eq.Objective} gives the upper bound $U_B^G \geq \lim_{K \rightarrow \infty}\mathbb{E}[J_K^G]$, where 
\begin{equation}\label{eq.performance_Greedy_M}
U_B^G = \frac{1}{2M}\left(\sum_{i=1}^M\alpha_i\right)\left(\frac{1}{T}\sum_{i=1}^M\frac{1}{p_i}-1\right)\mbox{Y}^G+\frac{1}{2M}\left(\sum_{i=1}^M\alpha_i\right) \; ,
\end{equation}
and
\begin{equation}
\mbox{Y}^G = 1+\frac{1}{M}+\frac{\displaystyle4-\frac{1}{T}+\frac{2}{M}}{\displaystyle\frac{1}{T}\sum_{i=1}^M\frac{1}{p_i}-1}+\frac{\displaystyle4-\frac{1}{T}+\frac{1}{M}+\frac{M}{T^2}\bar{\mathbb{V}}\left[\frac{1}{p_i}\right]}{\displaystyle\left[\frac{1}{T}\sum_{i=1}^M\frac{1}{p_i}-1\right]^2} \; .
\end{equation}
A more insightful expression for $U_B^G$ can be obtained by assuming that $M\rightarrow\infty$, as follows
\begin{equation}\label{eq.UBG_highM}
U_B^G = \displaystyle\frac{1}{2MT}\left(\sum_{i=1}^M\alpha_i\right)\left(\sum_{i=1}^M \frac{1}{p_i}\right)\left[1+\frac{C_V^2}{M}\right]+\frac{1}{2M}\left(\sum_{i=1}^M\alpha_i\right) \; .
\end{equation}
Dividing \eqref{eq.UBG_highM} by the lower bound in \eqref{eq.LowerBound} yields the performance guarantee $\rho^G$ in \eqref{eq.performance_Greedy}.
\end{proof}
\section{Proof of Theorem \ref{theo.performance_Random}}\label{app.Theo_performance_Random}
\noindent \textbf{Theorem \ref{theo.performance_Random}} (Performance of Randomized). Consider a network $(M,T,p_i,\alpha_i)$ with an infinite time-horizon. The Randomized policy with positive values of $\beta_i$ is $\rho^R$-optimal, where
\begin{equation}
\rho^R=2\frac{\displaystyle\left(\sum_{j=1}^{M}\beta_j\sum_{i=1}^{M}\frac{\alpha_i }{p_i\beta_i}\right)+(T-1)\left(\sum_{i=1}^{M}\frac{\alpha_i}{p_i}\right)}{\displaystyle\left(\sum_{i=1}^M \sqrt{\frac{\alpha_i}{p_i}}\right)^2+T\left(\sum_{i=1}^M\alpha_i\right)} \; .
\end{equation}

\begin{proof}
Consider the objective function for the Randomized policy in \eqref{eq.Objective_Random} and assume a general frame length $T$. To find an upper bound $U_B^R$, we derive a lower bound on $\mathbb{E}[d_i]$. Recall that during frame $k$, the Randomized policy can select the same client multiple times. Hence, the probability of delivering a packet to client $i$ during frame $k$ is given by
\begin{align}\label{eq.prob_Random}
\mathbb{E}&\left[d_i\right]=\mathbb{P}(\mbox{delivery to client $i$ during frame $k$}) \nonumber\\
&=\sum_{s=0}^{T}\mathbb{P}(\,\mbox{delivery}\,|\,\mbox{s selections}\,)\mathbb{P}(\,\mbox{s selections}\,) \nonumber\\
&=\sum_{s=0}^{T}\left[1-(1-p_i)^s\right]\binom{T}{s}\left(\frac{\beta_i}{\sum_{j=1}^{M}\beta_j}\right)^s\left(1-\frac{\beta_i}{\sum_{j=1}^{M}\beta_j}\right)^{T-s} \nonumber\\
&\overset{(a)}{\geq} p_i\sum_{s=1}^{T}\binom{T}{s}\left(\frac{\beta_i}{\sum_{j=1}^{M}\beta_j}\right)^s\left(1-\frac{\beta_i}{\sum_{j=1}^{M}\beta_j}\right)^{T-s}\nonumber\\
&=p_i\left[1-\left(1-\frac{\beta_i}{\sum_{j=1}^{M}\beta_j}\right)^{T}\right] \nonumber\\
&\overset{(b)}{\geq} \frac{p_iT\beta_i}{\sum_{j=1}^{M}\beta_j+(T-1)\beta_i} \; ,
\end{align}
where (a) uses $p_i \leq 1-(1-p_i)^s$ for $s \in \{1,2,\cdots,T\}$ and (b) uses \cite[inequality $r_5$]{inequality} which is given below for convenience 
\begin{equation}
(1-x)^T \leq \frac{1}{\left(1+\displaystyle\frac{Tx}{1-x}\right)} \; , \mbox{ for } x \in (0,1) \mbox{ and } T \geq 1 \; .
\end{equation}

Substituting the lower bound \eqref{eq.prob_Random} into the objective function for the Randomized policy in \eqref{eq.Objective_Random} gives 
\begin{align}\label{eq.Objective_Random_UB}
\lim_{K \rightarrow \infty}\mathbb{E}\left[J_K^{R}\right] &=\frac{1}{M}\sum_{i=1}^{M}\frac{\alpha_i}{\mathbb{E}\left[d_i\right]} \leq U_B^R \; ,
\end{align}
where the upper bound is given by 
\begin{equation}\label{eq.UB_Random}
U_B^R = \frac{1}{TM}\sum_{j=1}^{M}\beta_j\sum_{i=1}^{M}\frac{\alpha_i }{p_i\beta_i}+\frac{T-1}{TM}\sum_{i=1}^{M}\frac{\alpha_i}{p_i} \; .
\end{equation}
Finally, dividing \eqref{eq.UB_Random} by the lower bound in \eqref{eq.LowerBound} yields the performance guarantee $\rho^R$ in \eqref{eq.performance_Random}.
\end{proof}
\section{Proof of Theorem \ref{theo.performance_MaxWeight}}\label{app.Theo_performance_MaxWeight}
\noindent \textbf{Theorem \ref{theo.performance_MaxWeight}} (Performance of Max-Weight). Consider a network $(M,T,p_i,\alpha_i)$ with an infinite time-horizon. The Max-Weight policy is $\rho^{MW}$-optimal, where
\begin{equation}
\rho^{MW}=4 \frac{\displaystyle\left(\sum_{i=1}^{M}\sqrt{\frac{\alpha_i}{p_i}}\right)^2+(T-1)\sum_{i=1}^M\frac{\alpha_i}{p_i}}{\displaystyle\left(\sum_{i=1}^{M}\sqrt{\frac{\alpha_i}{p_i}}\right)^2+T\left(\sum_{i=1}^M\alpha_i\right)} \; .
\end{equation}

\begin{proof}
To obtain the upper bound, $\lim_{K \rightarrow \infty}\mathbb{E}\left[J_K^{MW}\right] \leq U_B^{MW}$, we manipulate the expression of the one-frame Lyapunov drift. Since $\Delta(\vec{h}_k)$ is central to this proof, we rewrite \eqref{eq.Lyapunov_drift} below for convenience 
\begin{align*}
\Delta(\vec{h}_k)=&-\frac{1}{M}\sum_{i=1}^M\mathbb{E}\left[d_i(k)|\vec{h}_{k}\right]\alpha_ih_{k,i}(h_{k,i}+2)+\nonumber\\
&+\frac{2}{M}\sum_{i=1}^M\alpha_ih_{k,i}+\frac{1}{M}\sum_{i=1}^M\alpha_i \; .
\end{align*}
Recall that the Max-Weight Policy minimizes $\Delta(\vec{h}_k)$ by choosing $\mathbb{E}\left[d_i(k)|\vec{h}_{k}\right]$ such that the sum 
$$
\sum_{i=1}^M\mathbb{E}\left[d_i(k)|\vec{h}_{k}\right]\alpha_ih_{k,i}(h_{k,i}+2)
$$
is maximized. Employing any other policy $\pi \in \Pi$ yields a lower (or equal) sum. 
Consider a Randomized Policy as defined in Sec.~\ref{sec.Randomized}. Its expected throughput is constant in every frame $k$, independently of the value of $\vec{h}_k$, thus $\mathbb{E}\left[d_i(k)|\vec{h}_{k}\right]=\mathbb{E}\left[d_i\right]$. Substituting $\mathbb{E}\left[d_i\right]$ into the equation of the one-frame Lyapunov Drift gives 
\begin{align}\label{eq.drift_first}
\Delta(\vec{h}_k)\leq&-\frac{1}{M}\sum_{i=1}^M\mathbb{E}\left[d_i\right]\alpha_ih_{k,i}(h_{k,i}+2)+\frac{2}{M}\sum_{i=1}^M\alpha_ih_{k,i}+\frac{1}{M}\sum_{i=1}^M\alpha_i \nonumber\\
=&-\frac{1}{M}\left\{\sum_{i=1}^M \alpha_i\mathbb{E}\left[d_i\right]\left(h_{k,i}-\frac{1}{\mathbb{E}\left[d_i\right]}+1\right)^2\right\}+\nonumber\\
&+\frac{1}{M}\sum_{i=1}^M \alpha_i\left[\mathbb{E}\left[d_i\right]\left(\frac{1}{\mathbb{E}\left[d_i\right]}-1\right)^2+1\right] \; .
\end{align}
Consider the Cauchy-Schwarz inequality 
\begin{align*}
&\left\{\sum_{i=1}^M \alpha_i\mathbb{E}\left[d_i\right]\left(h_{k,i}-\frac{1}{\mathbb{E}\left[d_i\right]}+1\right)^2\right\}\left\{\sum_{i=1}^M \frac{\alpha_i}{\mathbb{E}\left[d_i\right]}\right\} \geq \\
&\hspace{2cm}\geq \left\{\sum_{i=1}^M\alpha_i\left|h_{k,i}-\frac{1}{\mathbb{E}\left[d_i\right]}+1\right|\right\}^2 \; .
\end{align*}
Applying this inequality to \eqref{eq.drift_first} gives 
\begin{align*}
\Delta(\vec{h}_k)\leq&-\frac{1}{M}\left\{\sum_{i=1}^M \frac{\alpha_i}{\mathbb{E}\left[d_i\right]}\right\}^{-1}\left\{\sum_{i=1}^M\alpha_i\left|h_{k,i}-\frac{1}{\mathbb{E}\left[d_i\right]}+1\right|\right\}^2+\\
&+\frac{1}{M}\sum_{i=1}^M \alpha_i\left[\mathbb{E}\left[d_i\right]\left(\frac{1}{\mathbb{E}\left[d_i\right]}-1\right)^2+1\right] \; ,
\end{align*}
and rearranging the terms 
\begin{align*}
\frac{1}{M}\left\{\sum_{i=1}^M\alpha_i\left|h_{k,i}-\frac{1}{\mathbb{E}\left[d_i\right]}+1\right|\right\}^2 \leq -\left\{\sum_{i=1}^M \frac{\alpha_i}{\mathbb{E}\left[d_i\right]}\right\}\Delta(\vec{h}_k)+\\+\left\{\sum_{i=1}^M \frac{\alpha_i}{\mathbb{E}\left[d_i\right]}\right\}\frac{1}{M}\sum_{i=1}^M \alpha_i\left[\mathbb{E}\left[d_i\right]\left(\frac{1}{\mathbb{E}\left[d_i\right]}-1\right)^2+1\right] \; .
\end{align*}

Now, taking the expectation with respect to $\vec{h}_k$ yields 
\begin{align*}
&\frac{1}{M}\mathbb{E}\left[\left\{\sum_{i=1}^M\alpha_i\left|h_{k,i}-\frac{1}{\mathbb{E}\left[d_i\right]}+1\right|\right\}^2\right] \leq \\
&\quad\quad\leq -\left\{\sum_{i=1}^M \frac{\alpha_i}{\mathbb{E}\left[d_i\right]}\right\}\mathbb{E}\left[\Delta(\vec{h}_k)\right]+\\
&\quad\quad+\left\{\sum_{i=1}^M \frac{\alpha_i}{\mathbb{E}\left[d_i\right]}\right\}\frac{1}{M}\sum_{i=1}^M \alpha_i\left[\mathbb{E}\left[d_i\right]\left(\frac{1}{\mathbb{E}\left[d_i\right]}-1\right)^2+1\right] \; ,
\end{align*}
summing over $k \in \{1,2,\cdots,K\}$ and dividing by $K$ results in 
\begin{align}
&\frac{1}{KM}\sum_{k=1}^K\mathbb{E}\left[\left\{\sum_{i=1}^M\alpha_i\left|h_{k,i}-\frac{1}{\mathbb{E}\left[d_i\right]}+1\right|\right\}^2\right] \leq \label{eq.drift_complete}\\
&\quad\quad\leq -\left\{\sum_{i=1}^M \frac{\alpha_i}{\mathbb{E}\left[d_i\right]}\right\}\frac{1}{K}\sum_{k=1}^K\mathbb{E}\left[\Delta(\vec{h}_k)\right]+\nonumber\\
&\quad\quad+\left\{\sum_{i=1}^M \frac{\alpha_i}{\mathbb{E}\left[d_i\right]}\right\}\frac{1}{M}\sum_{i=1}^M \alpha_i\left[\mathbb{E}\left[d_i\right]\left(\frac{1}{\mathbb{E}\left[d_i\right]}-1\right)^2+1\right] \; .\nonumber
\end{align}

For simplicity of exposition, we divide inequality \eqref{eq.drift_complete} in two terms LHS $\leq$ RHS, analyzing each part separately. Applying Jensen's inequality to the LHS twice, gives 
\begin{align*}
\frac{1}{KM}\sum_{k=1}^K\mathbb{E}\left[\left\{\sum_{i=1}^M\alpha_i\left|h_{k,i}-\frac{1}{\mathbb{E}\left[d_i\right]}+1\right|\right\}^2\right] &\leq \mbox{RHS}\; ;\\ 
\frac{1}{KM}\sum_{k=1}^K\mathbb{E}\left[\sum_{i=1}^M\alpha_i\left|h_{k,i}-\frac{1}{\mathbb{E}\left[d_i\right]}+1\right|\right]^2 &\leq \mbox{RHS} \; ;\\ 
\frac{1}{M}\left\{\frac{1}{K}\sum_{k=1}^K\mathbb{E}\left[\sum_{i=1}^M\alpha_i\left|h_{k,i}-\frac{1}{\mathbb{E}\left[d_i\right]}+1\right|\right]\right\}^2 &\leq \mbox{RHS} \; .
\end{align*}
Then, by further manipulating this expression, we have 
\begin{align}\label{eq.drift_LHS}
\frac{1}{M}&\left|\frac{1}{K}\sum_{k=1}^K\mathbb{E}\left[\sum_{i=1}^M\alpha_i\left|h_{k,i}-\frac{1}{\mathbb{E}\left[d_i\right]}+1\right|\right]\right| \leq \sqrt{\frac{\mbox{RHS}}{M}} \; ;\nonumber\\ 
\frac{1}{KM}&\sum_{k=1}^K\sum_{i=1}^M\mathbb{E}\left[\alpha_i\left(h_{k,i}-\frac{1}{\mathbb{E}\left[d_i\right]}+1\right)\right] \leq \sqrt{\frac{\mbox{RHS}}{M}} \; ;\nonumber\\
\frac{1}{KM}&\sum_{k=1}^K\sum_{i=1}^M\alpha_i\mathbb{E}\left[h_{k,i}\right] \leq \frac{1}{M}\sum_{i=1}^M\alpha_i\left(\frac{1}{\mathbb{E}\left[d_i\right]}-1\right) + \sqrt{\frac{\mbox{RHS}}{M}}\; ;\nonumber\\
&\mathbb{E}\left[J_K^{MW}\right] \leq \frac{1}{M}\sum_{i=1}^M\frac{\alpha_i}{\mathbb{E}\left[d_i\right]} + \sqrt{\frac{\mbox{RHS}}{M}} \; .
\end{align}

Going back to \eqref{eq.drift_complete} and analyzing the first term on the RHS gives 
\begin{align*}
-\frac{1}{K}\sum_{k=1}^K\mathbb{E}\left[\Delta(\vec{h}_k)\right]&=\frac{1}{K}\left\{\mathbb{E}\left[L(\vec{h}_1)\right]-\mathbb{E}\left[L(\vec{h}_{K+1})\right]\right\} \; ;\\
-\frac{1}{K}\sum_{k=1}^K\mathbb{E}\left[\Delta(\vec{h}_k)\right]&\leq\frac{\mathbb{E}[L(\vec{h}_1)]}{K} \; ,
\end{align*}
and the second term on the RHS is such that 
\begin{align*}
\frac{1}{M}\sum_{i=1}^M\alpha_i\left[\mathbb{E}\left[d_i\right]\left(\frac{1}{\mathbb{E}\left[d_i\right]}-1\right)^2+1\right] \leq \frac{1}{M}\sum_{i=1}^M\left[\frac{\alpha_i}{\mathbb{E}\left[d_i\right]}\right] \; .
\end{align*}
Using both results, the RHS of \eqref{eq.drift_complete} can be upper bounded by 
\begin{equation*}
RHS \leq \left\{\sum_{i=1}^M \frac{\alpha_i}{\mathbb{E}\left[d_i\right]}\right\}\left\{\frac{\mathbb{E}[L(\vec{h}_1)]}{K}+\frac{1}{M}\sum_{i=1}^M \frac{\alpha_i}{\mathbb{E}\left[d_i\right]}\right\} \; ,
\end{equation*}
and, in the limit $K \rightarrow \infty$ 
\begin{equation*}
RHS \leq \left\{\sum_{i=1}^M \frac{\alpha_i}{\mathbb{E}\left[d_i\right]}\right\}\left\{\frac{1}{M}\sum_{i=1}^M \frac{\alpha_i}{\mathbb{E}\left[d_i\right]}\right\}=\frac{1}{M}\left\{\sum_{i=1}^M \frac{\alpha_i}{\mathbb{E}\left[d_i\right]}\right\}^2 \; .
\end{equation*}

Substituting the upper bound for the RHS into the inequality \eqref{eq.drift_LHS} and applying the limit $K \rightarrow \infty$, gives 
\begin{align}\label{eq.drift_UB1}
\lim_{K\rightarrow\infty}\mathbb{E}\left[J_K^{MW}\right] &\leq \frac{1}{M}\sum_{i=1}^M\frac{\alpha_i}{\mathbb{E}\left[d_i\right]}+ \sqrt{\frac{\mbox{RHS}}{M}} \nonumber\\
&\leq \frac{1}{M}\sum_{i=1}^M\frac{\alpha_i}{\mathbb{E}\left[d_i\right]}+ \frac{1}{M}\sum_{i=1}^M\frac{\alpha_i}{\mathbb{E}\left[d_i\right]} \nonumber\\
&=\frac{2}{M}\sum_{i=1}^M\frac{\alpha_i}{\mathbb{E}\left[d_i\right]} \; .
\end{align}

The expression in \eqref{eq.drift_UB1} gives an upper bound to the performance of the Max-Weight Policy as a function of $\mathbb{E}\left[d_i\right]$. Using inequality \eqref{eq.prob_Random} from Appendix \ref{app.Theo_performance_Random}, namely 
\begin{equation*}
\mathbb{E}\left[d_i\right] \geq \frac{p_iT\beta_i}{\sum_{j=1}^{M}\beta_j+(T-1)\beta_i} \; ,
\end{equation*}
results in the following upper bound 
\begin{align*}
\lim_{K\rightarrow\infty}\mathbb{E}\left[J_K^{MW}\right] &\leq \frac{2}{MT}\sum_{j=1}^{M}\beta_j\sum_{i=1}^M\frac{\alpha_i}{p_i\beta_i}+\frac{2(T-1)}{MT}\sum_{i=1}^M\frac{\alpha_i}{p_i} \; ,
\end{align*}
where $\beta_i/\sum_{j=1}^M\beta_j$ is the probability of the Randomized Policy selecting client $i$ for transmission in any given slot. To obtain an upper bound that is a function of only the network setup, consider the Randomized Policy described in Corollary \ref{cor.performance_Random}, which assigns $\beta_i^2=\alpha_i/p_i$. Using this assignment, we obtain the upper bound $\lim_{K\rightarrow\infty}\mathbb{E}\left[J_K^{MW}\right] \leq U_B^{MW}$, where 
\begin{equation}\label{eq.UB_MaxWeight}
U_B^{MW}=\frac{2}{MT}\left\{\left(\sum_{i=1}^{M}\sqrt{\frac{\alpha_i}{p_i}}\right)^2+(T-1)\sum_{i=1}^M\frac{\alpha_i}{p_i}\right\} \; .
\end{equation}
Finally, dividing \eqref{eq.UB_MaxWeight} by the lower bound in \eqref{eq.LowerBound} yields the performance guarantee $\rho^{MW}$ in \eqref{eq.performance_MaxWeight}.
\end{proof}
\section{Proof of Proposition \ref{prop.Threshold}}\label{app.Prop_Threshold}
\noindent \textbf{Proposition \ref{prop.Threshold}} (Threshold Policy). Consider the Frame-Based Decoupled Model over an infinite-horizon. The stationary scheduling policy $\pi^*$ that solves Bellman equations \eqref{eq.Bellman} is a threshold policy in which the BS transmits during frames that have $h \geq H$ and idles when $1\leq h<H$, where the threshold $H$ is given by
\begin{equation}
H=\left\lfloor 1-Z + \sqrt{Z^2+\frac{2C}{pT\alpha}} \right\rfloor \; ,
\end{equation}
and the value of $Z$ is
\begin{equation}
Z=\frac{1}{2}+\frac{(1-p)^T}{(1-(1-p)^T)} \; .
\end{equation}

\begin{proof}
During frame $k$, the scheduling policy must decide between transmitting and idling. If $\pi$ transmits, the value of $h$ may be reduced to $h=1$ and the network incurs an expected service charge of $\hat{C}=C(1-(1-p)^T)/p$. On the other hand, if $\pi$ idles, the value of $h$ is incremented by $1$ and there is no service charge. Intuitively, we expect that the optimal scheduling decision is to transmit during frames in which $h$ is high and idle when $h$ is low. In particular, if the optimal scheduling decision is to transmit when $h=H$, we expect that it is also optimal to transmit for all $h\geq H$. This behavior is characteristic of threshold policies.

In this appendix, we assume that $\pi^*$ is a threshold policy that idles when $1 \leq h <H$ and transmits when $h\geq H$, for a given value of $H\geq 1$. Using this assumption, we solve Bellman equations \eqref{eq.Bellman} and then show that the solution is consistent with the assumption. For convenience, we rewrite Bellman equations below as $S(1)=0$ and 
\begin{align}\label{eq.Bellman_2}
S(h)= & S(h+1)-\lambda T+T \alpha h + \nonumber\\
& +\min\left\{ 0 ; \hat{C} - \left[1-(1-p)^T\;\right]S(h+1) \right\} \; .
\end{align}

First, we \textbf{analyze the case} $\mathbf{h\geq H}$. According to \eqref{eq.Bellman_2}, the condition for the threshold policy $\pi^*$ to transmit in a frame with state $h$ is
\begin{equation}\label{eq.Condition_2}
S(h+1)>\frac{\hat{C}}{1-(1-p)^T} \;\mbox{ , for } h\geq H \; .
\end{equation}
Assuming that condition \eqref{eq.Condition_2} holds, it follows from \eqref{eq.Bellman_2} that
\begin{align*}
S(h)=-\lambda T+T \alpha h + \hat{C} +(1-p)^TS(h+1)\; .
\end{align*}
Since this expression is valid for all $h\geq H$, we can substitute $S(h+1)$ above and get
\begin{align*}
S(h)=&-\lambda T+T \alpha h + \hat{C} + \\
&+(1-p)^T\left[-\lambda T+T \alpha (h+1) + \hat{C}\right] +\\
&+(1-p)^{2T}S(h+2)\; .
\end{align*}
Repeating this procedure $n$ times, yields
\begin{align*}
S(h)=&[-\lambda T+T \alpha h + \hat{C}][1+(1-p)^T+\cdots+(1-p)^{nT}] + \\
&+T \alpha[(1-p)^{T}+2(1-p)^{2T}+\cdots+n(1-p)^{nT}] + \\
&+(1-p)^{(n+1)T}S(h+n+1)\; ,
\end{align*}
and in the limit $n \rightarrow \infty$ we have
\begin{equation*}
S(h)=\frac{T \alpha h + \hat{C}-\lambda T}{1-(1-p)^T} +\frac{T \alpha (1-p)^{T}}{\left(1-(1-p)^{T}\right)^2} \; .
\end{equation*}
Notice that $(1-p)^{(n+1)T}S(h+n+1) \rightarrow 0$ when $n \rightarrow \infty$. To emphasize that this expression is valid only for $h\geq H$, we substitute $h=H+j^+$ with $j^+\in\{0,1,2,\cdots\}$ and get
\begin{equation}\label{eq.S(H+)}
S(H+j^+)=\frac{T \alpha (H +\; j^+) + \hat{C}-\lambda T}{1-(1-p)^T} +\frac{T \alpha (1-p)^{T}}{\left(1-(1-p)^{T}\right)^2} \; .
\end{equation}

Next, we \textbf{analyze the case} $\mathbf{1\leq h < H}$. According to \eqref{eq.Bellman_2}, the condition for the threshold policy $\pi^*$ to idle in a frame with state $h$ is
\begin{equation}\label{eq.Condition_3}
S(h+1)\leq \frac{\hat{C}}{1-(1-p)^T} \;\mbox{ , for } 1\leq h < H \; .
\end{equation}
Assuming that condition \eqref{eq.Condition_3} holds, it follows from \eqref{eq.Bellman_2} that
\begin{align*}
S(h)=S(h+1)-\lambda T+T \alpha h \; .
\end{align*}
Since this expression is valid for $h \in \{1,2,\cdots,H-1\}$ and $S(H)$ is known from \eqref{eq.S(H+)}, we have
\begin{equation*}
S(H-1)=S(H)-\lambda T+T \alpha (H-1) \; .
\end{equation*}
Moreover,
\begin{align*}
S(H-2)&=S(H-1)-\lambda T+T \alpha (H-2) \\
	  &=S(H)-2\lambda T+2T \alpha H-T \alpha (1+2) \; .
\end{align*}
Repeating this procedure $n$ times, yields
\begin{align*}
S(H-n)&=S(H)-n\lambda T+nT \alpha H-T \alpha (1+2+\cdots+n) \; . \\
	  &=S(H)-n\lambda T+nT \alpha H-\frac{T \alpha (1+n)n}{2} \; . 
\end{align*}
To emphasize that this expression is valid only for $1 \leq h < H$, we substitute $h=H+j^-$, where $j^-\in\{-H+1,\cdots,-2,-1\}$ and get
\begin{equation}\label{eq.S(H-)}
S(H+j^-)=S(H)+j^-\left[\lambda T-T \alpha H-\frac{T \alpha (j^--1)}{2}\right] \; . 
\end{equation}

Expressions \eqref{eq.S(H+)} and \eqref{eq.S(H-)} give the differential cost-to-go $S(h)$ as a function of the threshold $H$ and the optimal average cost $\lambda$. To find both variables, $H$ and $\lambda$, we first analyze the optimal policy $\pi^*$ in the vicinity of the threshold. Policy $\pi^*$ idles when $h=H-1$ and transmits when $h=H$. Merging conditions \eqref{eq.Condition_2} and \eqref{eq.Condition_3} give:
\begin{align*}
S(H)\leq \frac{\hat{C}}{1-(1-p)^T} <S(H+1) \; .
\end{align*}
Since the expression for $S(H+j^+)$ in \eqref{eq.S(H+)} is monotonically increasing in $j^+ \in \{0,1,2,\cdots\}$, it follows that there exists $H+\gamma$ with $H \in \{1,2,3,\cdots\}$ and $\gamma \in [0,1)$ such that
\begin{equation}\label{eq.lambda_1}
S(H+\gamma)=\frac{\hat{C}}{1-(1-p)^T} \; .
\end{equation}
Substituting \eqref{eq.S(H+)} into \eqref{eq.lambda_1} yields
\begin{align}\label{eq.lambda_2}
&\frac{T \alpha (H+\gamma) -\lambda T}{1-(1-p)^T} +\frac{T \alpha (1-p)^{T}}{\left(1-(1-p)^{T}\right)^2}=0 \nonumber\\
&\lambda  T - T\alpha H= T\alpha\gamma+\frac{T\alpha(1-p)^{T}}{1-(1-p)^{T}} \; .
\end{align}

Next, we analyze the Bellman equation $S(1)=0$ using the expression for $S(H+j^-)$ in \eqref{eq.S(H-)} with $j^- = -H+1$, as follows
\begin{align*}
S(H)+(-H+1)&\left[\lambda T -T\alpha H + \frac{T \alpha H}{2}\right]=0 \; .
\end{align*}
Substituting $S(H)$ from \eqref{eq.S(H+)} gives
\begin{align}\label{eq.aux_2}
\frac{\hat{C}+ T \alpha H -\lambda T}{1-(1-p)^T} +\frac{T \alpha (1-p)^{T}}{\left(1-(1-p)^{T}\right)^2}=\nonumber\\
=(H-1)\left[\lambda T-T \alpha H+\frac{T \alpha H}{2}\right] \; .
\end{align}
Combining \eqref{eq.aux_2} and \eqref{eq.lambda_2} yields
\begin{align*}
\frac{\hat{C}-T \alpha \gamma}{1-(1-p)^T} =(H-1)\left[T\alpha\gamma+\frac{T\alpha(1-p)^{T}}{1-(1-p)^{T}}+\frac{T \alpha H}{2}\right] \; .
\end{align*}

Manipulating this quadratic equation on $H$ gives the unique positive solution:
\begin{align}\label{eq.H_proof_1}
H=\left(1-\gamma \right)-Z+\sqrt{\frac{2C}{T\alpha p}-\gamma\left(1-\gamma\right)+Z^2} \; ,
\end{align}
where 
\begin{align*}
\hat{C}=\frac{C(1-(1-p)^T)}{p} \; \mbox{ and } \; Z=\frac{1}{2}+\frac{(1-p)^T}{(1-(1-p)^T)} \; .
\end{align*}
It is easy to see from \eqref{eq.H_proof_1} that the derivative $dH/d\gamma<0$ when $\gamma \in [0,1)$, implying that $H$ is monotonically decreasing. Hence, in the range $\gamma \in [0,1)$, the value of $H$ decreases from
\begin{align*}
H(0)=1-Z+\sqrt{\frac{2C}{T\alpha p}+Z^2} \; \mbox{ to } \; H(1)=-Z+\sqrt{\frac{2C}{T\alpha p}+Z^2}
\end{align*}
Since $H(0)-H(1)=1$, there exists a unique $\gamma^* \in [0,1)$ such that $H(\gamma^*)$ is integer-valued and the expression for $H$ can be obtained as $H= H(\gamma^*)=\lfloor H(0) \rfloor$, or more explicitly by
\begin{equation}\label{eq.H_proof}
H=\left\lfloor 1-Z+\sqrt{Z^2+\frac{2C}{T\alpha p}} \right\rfloor \; .
\end{equation}

With the expression for $H$, we can obtain the optimal average cost per frame by isolating $\lambda$ in \eqref{eq.aux_2} as follows
\begin{equation}\label{eq.L_proof}
\lambda = \frac{\alpha}{1-(1-p)^T}+\frac{\displaystyle\frac{C}{Tp}+\frac{\alpha H (H-1)}{2}}{H+\displaystyle\frac{(1-p)^T}{1-(1-p)^T}} \; .
\end{equation}

Finally, with the closed-form expressions for the differential cost-to-go $S(h)$, threshold $H$ and optimal average cost per frame $\lambda$, it is possible to evaluate the consistency between the solution and the assumption of a threshold policy. For the solution of the Bellman equation \eqref{eq.Bellman_2} to be a threshold policy, the following condition must hold:
\begin{align*}
S(H+j^-+1) \leq \frac{\hat{C}}{1-(1-p)^T} < S(H+j^++1) \; ,
\end{align*}
for all $j^- \in \{-H+1,\cdots,-1\}$ and $j^+ \in \{0,1,\cdots\}$. Since $S(H+j^-)$ and $S(H+j^+)$ are monotonically increasing with $j^-$ and $j^+$, respectively, it is sufficient to show that
\begin{align}\label{eq.condition_S}
S(H)\leq \frac{\hat{C}}{1-(1-p)^T} < S(H+1) \; .
\end{align}

Recall from \eqref{eq.lambda_1} that there exists $\gamma \in [0,1)$ such that
\begin{equation*}
S(H+\gamma)=\frac{\hat{C}}{1-(1-p)^T} \; .
\end{equation*}
From the monotonicity of $S(H+j^+)$, it follows that condition \eqref{eq.condition_S} is satisfied. Thus, the solution to Bellman equations is consistent.
\end{proof}
\section{Proof of Theorem \ref{theo.performance_Whittle}}\label{app.Theo_performance_Whittle}
\noindent \textbf{Theorem \ref{theo.performance_Whittle}} (Performance of Whittle). Consider a network $(M,T,p_i,\alpha_i)$ with an infinite time-horizon. The Whittle's Index policy is $\rho^{WI}$-optimal, where
\begin{equation}
\rho^{WI}=4 \frac{\displaystyle\left(\sum_{i=1}^{M}\sqrt{\frac{\mathbf{\widetilde{\alpha}_i}}{p_i}}\right)^2+(T-1)\sum_{i=1}^M\frac{\mathbf{\widetilde{\alpha}_i}}{p_i}}{\displaystyle\left(\sum_{i=1}^{M}\sqrt{\frac{\alpha_i}{p_i}}\right)^2+T\left(\sum_{i=1}^M\alpha_i\right)} \; ,
\end{equation}
and 
\begin{equation}
\mathbf{\widetilde{\alpha}_i} = \frac{\alpha_i}{2}\left(\frac{2}{1-(1-p_i)^T}+1\right)^2\; .
\end{equation}

\begin{proof}
From the definition of the Whittle Index Policy, it can be seen that the choice of $\mathbb{E}\left[d_i(k)|\vec{h}_{k}\right]$ is such that the sum 
$$
\sum_{i=1}^M\mathbb{E}\left[d_i(k)|\vec{h}_{k}\right]\alpha_ih_{k,i}\left(h_{k,i}+\frac{1+(1-p_i)^T}{1-(1-p_i)^T}\right)
$$
is maximized. Notice that the difference between Whittle and Max-Weight is only the last term in the sum. Denote this term by
$$
Y_i=\frac{1+(1-p_i)^T}{1-(1-p_i)^T} \; .
$$

The first step is to find an upper bound to the one-frame Lyapunov drift $\Delta(\vec{h}_k)$ that has the Whittle Index Policy as its minimizer. If this can be done, the subsequent arguments of the proof are analogous to the ones for the derivation of the performance guarantee of the Max-Weight Policy, $\rho^{MW}$, in Appendix \ref{app.Theo_performance_MaxWeight}. Consider the expression of $\Delta(\vec{h}_k)$ in \eqref{eq.Lyapunov_drift} stated below
\begin{align*}
\Delta(\vec{h}_k)=&-\frac{1}{M}\sum_{i=1}^M\mathbb{E}\left[d_i(k)|\vec{h}_{k}\right]\alpha_ih_{k,i}(h_{k,i}+2)+\nonumber\\
&+\frac{2}{M}\sum_{i=1}^M\alpha_ih_{k,i}+\frac{1}{M}\sum_{i=1}^M\alpha_i \; . \\
\end{align*}

Manipulating the first term on the RHS of $\Delta(\vec{h}_k)$, yields
\begin{align}\label{eq.WhittleBound_aux}
-&\sum_{i=1}^M\mathbb{E}\left[d_i(k)|\vec{h}_{k}\right]\alpha_ih_{k,i}(h_{k,i}+2)=\nonumber\\
=&-\sum_{i=1}^M\mathbb{E}\left[d_i(k)|\vec{h}_{k}\right]\alpha_ih_{k,i}(h_{k,i}+Y_i)+\nonumber\\
&+\sum_{i=1}^M\mathbb{E}\left[d_i(k)|\vec{h}_{k}\right]\alpha_ih_{k,i}Y_i-\sum_{i=1}^M\mathbb{E}\left[d_i(k)|\vec{h}_{k}\right]\alpha_ih_{k,i}2 \nonumber\\
\leq&-\sum_{i=1}^M\mathbb{E}\left[d_i(k)|\vec{h}_{k}\right]\alpha_ih_{k,i}(h_{k,i}+Y_i)+\sum_{i=1}^M\alpha_ih_{k,i}Y_i \; .
\end{align}

Substituting \eqref{eq.WhittleBound_aux} into the expression of $\Delta(\vec{h}_k)$ in \eqref{eq.Lyapunov_drift} gives
\begin{align}\label{eq.WhittleBound}
\Delta(\vec{h}_k)\leq &-\frac{1}{M}\sum_{i=1}^M\mathbb{E}\left[d_i(k)|\vec{h}_{k}\right]\alpha_ih_{k,i}(h_{k,i}+Y_i)+\nonumber\\
&+\frac{1}{M}\sum_{i=1}^M\alpha_ih_{k,i}(2+Y_i)+\frac{1}{M}\sum_{i=1}^M\alpha_i \; . 
\end{align}
Observe that the Whittle Index Policy minimizes this upper bound. Hence, employing any other policy $\pi \in \Pi$ yields a higher (or equal) bound. 
Consider a Randomized Policy as defined in Sec.~\ref{sec.Randomized}. Its expected throughput is constant in every frame $k$, independently of the value of $\vec{h}_k$, thus $\mathbb{E}\left[d_i(k)|\vec{h}_{k}\right]=\mathbb{E}\left[d_i\right]$. Substituting $\mathbb{E}\left[d_i\right]$ into the upper bound gives
\begin{align}\label{eq.drift_first_Whittle}
\Delta(\vec{h}_k)\leq &-\frac{1}{M}\sum_{i=1}^M\mathbb{E}\left[d_i\right]\alpha_ih_{k,i}(h_{k,i}+Y_i)+\nonumber\\
&+\frac{1}{M}\sum_{i=1}^M\alpha_ih_{k,i}(2+Y_i)+\frac{1}{M}\sum_{i=1}^M\alpha_i \; ;\nonumber\\
\Delta(\vec{h}_k)\leq &-\frac{1}{M}\left\{\sum_{i=1}^M \alpha_i\mathbb{E}\left[d_i\right]\left(h_{k,i}-\frac{2+Y_i}{2\mathbb{E}\left[d_i\right]}+\frac{Y_i}{2}\right)^2\right\}+\nonumber\\
&+\frac{1}{M}\sum_{i=1}^M \alpha_i\left[\mathbb{E}\left[d_i\right]\left(\frac{2+Y_i}{2\mathbb{E}\left[d_i\right]}-\frac{Y_i}{2}\right)^2+1\right] \; .
\end{align}
Consider the Cauchy-Schwarz inequality
\begin{align*}
&\left\{\sum_{i=1}^M \alpha_i\mathbb{E}\left[d_i\right]\left(h_{k,i}-\frac{2+Y_i}{2\mathbb{E}\left[d_i\right]}+\frac{Y_i}{2}\right)^2\right\}\left\{\sum_{i=1}^M \frac{\alpha_i}{\mathbb{E}\left[d_i\right]}\right\} \geq \\
&\hspace{2cm}\geq \left\{\sum_{i=1}^M\alpha_i\left|h_{k,i}-\frac{2+Y_i}{2\mathbb{E}\left[d_i\right]}+\frac{Y_i}{2}\right|\right\}^2 \; .
\end{align*}
Applying this inequality to \eqref{eq.drift_first_Whittle} gives
\begin{align*}
\Delta(\vec{h}_k)\leq&-\frac{1}{M}\left\{\sum_{i=1}^M \frac{\alpha_i}{\mathbb{E}\left[d_i\right]}\right\}^{-1}\left\{\sum_{i=1}^M\alpha_i\left|h_{k,i}-\frac{2+Y_i}{2\mathbb{E}\left[d_i\right]}+\frac{Y_i}{2}\right|\right\}^2\\
&+\frac{1}{M}\sum_{i=1}^M \alpha_i\left[\mathbb{E}\left[d_i\right]\left(\frac{2+Y_i}{2\mathbb{E}\left[d_i\right]}-\frac{Y_i}{2}\right)^2+1\right] \; .
\end{align*}

Now, rearranging the terms, taking expectation with respect to $\vec{h}_k$, summing over $k \in \{1,2,\cdots,K\}$ and then dividing by $K$ results in
\begin{align}
&\frac{1}{KM}\sum_{k=1}^K\mathbb{E}\left[\left\{\sum_{i=1}^M\alpha_i\left|h_{k,i}-\frac{2+Y_i}{2\mathbb{E}\left[d_i\right]}+\frac{Y_i}{2}\right|\right\}^2\right] \leq \label{eq.drift_complete_Whittle}\\
&\quad\quad\leq -\left\{\sum_{i=1}^M \frac{\alpha_i}{\mathbb{E}\left[d_i\right]}\right\}\frac{1}{K}\sum_{k=1}^K\mathbb{E}\left[\Delta(\vec{h}_k)\right]+\nonumber\\
&\quad\quad+\left\{\sum_{i=1}^M \frac{\alpha_i}{\mathbb{E}\left[d_i\right]}\right\}\frac{1}{M}\sum_{i=1}^M \alpha_i\left[\mathbb{E}\left[d_i\right]\left(\frac{2+Y_i}{2\mathbb{E}\left[d_i\right]}-\frac{Y_i}{2}\right)^2+1\right] \; .\nonumber
\end{align}

For simplicity of exposition, we divide inequality \eqref{eq.drift_complete_Whittle} in two terms LHS $\leq$ RHS and analyze each part separately. Applying Jensen's inequality to the LHS twice and then manipulating the resulting expression gives
\begin{align}\label{eq.drift_LHS_Whittle}
\frac{1}{M}&\left\{\frac{1}{K}\sum_{k=1}^K\mathbb{E}\left[\sum_{i=1}^M\alpha_i\left|h_{k,i}-\frac{2+Y_i}{2\mathbb{E}\left[d_i\right]}+\frac{Y_i}{2}\right|\right]\right\}^2 \leq \mbox{RHS} \; ; \nonumber\\
\frac{1}{KM}&\sum_{k=1}^K\sum_{i=1}^M\mathbb{E}\left[\alpha_i\left(h_{k,i}-\frac{2+Y_i}{2\mathbb{E}\left[d_i\right]}+\frac{Y_i}{2}\right)\right] \leq \sqrt{\frac{\mbox{RHS}}{M}} \; ; \nonumber\\
\frac{1}{KM}&\sum_{k=1}^K\sum_{i=1}^M\alpha_i\mathbb{E}\left[h_{k,i}\right]\leq \frac{1}{M}\sum_{i=1}^M\alpha_i\left(\frac{2+Y_i}{2\mathbb{E}\left[d_i\right]}-\frac{Y_i}{2}\right) + \sqrt{\frac{\mbox{RHS}}{M}} \; ; \nonumber\\
&\mathbb{E}\left[J_K^{WI}\right] \leq \frac{1}{M}\sum_{i=1}^M\alpha_i\left[\frac{\left(2+Y_i\right)^2}{2\mathbb{E}\left[d_i\right]}\right] + \sqrt{\frac{\mbox{RHS}}{M}} \; .
\end{align}

Going back to \eqref{eq.drift_complete_Whittle} and analyzing the first term on the RHS gives
\begin{align*}
-\frac{1}{K}\sum_{k=1}^K\mathbb{E}\left[\Delta(\vec{h}_k)\right]&\leq\frac{\mathbb{E}[L(\vec{h}_1)]}{K} \; ,
\end{align*}
and the second term on the RHS is such that
\begin{align*}
\sum_{i=1}^M\alpha_i\left[\mathbb{E}\left[d_i\right]\left(\frac{2+Y_i}{2\mathbb{E}\left[d_i\right]}-\frac{Y_i}{2}\right)^2+1\right] \leq \sum_{i=1}^M\alpha_i\left[\frac{\left(2+Y_i\right)^2}{2\mathbb{E}\left[d_i\right]}\right] \; .
\end{align*}
Using both results, the RHS of \eqref{eq.drift_complete_Whittle} can be upper bounded by
\begin{equation*}
RHS \leq \left\{\sum_{i=1}^M \frac{\alpha_i}{\mathbb{E}\left[d_i\right]}\right\}\left\{\frac{\mathbb{E}[L(\vec{h}_1)]}{K}+\frac{1}{M}\sum_{i=1}^M\alpha_i\left[\frac{\left(2+Y_i\right)^2}{2\mathbb{E}\left[d_i\right]}\right]\right\} \; ,
\end{equation*}
and, in the limit $K \rightarrow \infty$
\begin{align*}
RHS &\leq \frac{1}{M}\left\{\sum_{i=1}^M \frac{\alpha_i}{\mathbb{E}\left[d_i\right]}\right\}\left\{\sum_{i=1}^M\alpha_i\left[\frac{\left(2+Y_i\right)^2}{2\mathbb{E}\left[d_i\right]}\right]\right\} \\
&\leq \frac{1}{M}\left\{\sum_{i=1}^M\alpha_i\left[\frac{\left(2+Y_i\right)^2}{2\mathbb{E}\left[d_i\right]}\right]\right\}^2 \; .
\end{align*}

Substituting the upper bound for the RHS into the inequality \eqref{eq.drift_LHS_Whittle} and applying the limit $K \rightarrow \infty$, gives
\begin{align}\label{eq.drift_UB1_Whittle}
\lim_{K\rightarrow\infty}\mathbb{E}\left[J_K^{WI}\right] &\leq \frac{2}{M}\sum_{i=1}^M\alpha_i\left[\frac{\left(2+Y_i\right)^2}{2\mathbb{E}\left[d_i\right]}\right] \; .
\end{align}

The expression in \eqref{eq.drift_UB1_Whittle} gives an upper bound to the performance of the Whittle Index Policy as a function of $\mathbb{E}\left[d_i\right]$. Define the auxiliary variable $\widetilde{\alpha}_i$ as
\begin{equation*}
\mathbf{\widetilde{\alpha}_i} = \frac{\alpha_i}{2}\left(2+Y_i\right)^2\; .
\end{equation*}
Substituting $\widetilde{\alpha}_i$ into \eqref{eq.drift_UB1_Whittle} gives
\begin{align*}
\lim_{K\rightarrow\infty}\mathbb{E}\left[J_K^{WI}\right] &\leq \frac{2}{M}\sum_{i=1}^M\left[\frac{\mathbf{\widetilde{\alpha}_i}}{\mathbb{E}\left[d_i\right]}\right] \; ,
\end{align*}
which is identical to the upper bound for the Max-Weight Policy in \eqref{eq.drift_UB1}. Thus, using the same arguments, we obtain the upper bound $\lim_{K\rightarrow\infty}\mathbb{E}\left[J_K^{WI}\right] \leq U_B^{WI}$, where
\begin{equation}\label{eq.UB_Whittle}
U_B^{WI}=\frac{2}{MT}\left\{\left(\sum_{i=1}^{M}\sqrt{\frac{\mathbf{\widetilde{\alpha}_i}}{p_i}}\right)^2+(T-1)\sum_{i=1}^M\frac{\mathbf{\widetilde{\alpha}_i}}{p_i}\right\} \; .
\end{equation}
Finally, dividing \eqref{eq.UB_Whittle} by the lower bound in \eqref{eq.LowerBound} yields the performance guarantee $\rho^{WI}$ in \eqref{eq.performance_Whittle}.
\end{proof}
\end{document}